\def\journal{0}
\ifnum\journal=1
\RequirePackage{fix-cm}
\documentclass{svjour3}                     
\smartqed  
\else
\documentclass[11pt,letterpaper]{article}
\fi
\usepackage{graphicx}
\usepackage{amssymb}
\usepackage{amsmath}

\ifnum\journal=0
\usepackage{amsthm}
\usepackage[margin=.8in]{geometry}
\fi

\usepackage{color}
\usepackage{hyperref}
\usepackage{times}
\usepackage{verbatim}
\usepackage[h]{esvect}

\DeclareMathOperator{\IC}{IC}
\DeclareMathOperator{\icost}{ICost}
\newcommand{\eps}{\varepsilon}
\newcommand{\PC}{\mbox{\sf PC}}
\newcommand{\SC}{\mbox{\sf SC}}
\newcommand{\LEP}{\mbox{\sf LPCE}}
\newcommand{\good}{\mbox{\rm Good}}
\DeclareMathOperator{\opequal}{\mbox{\sf EQUAL}}
\DeclareMathOperator{\opor}{\mbox{\sf OR}}
\DeclareMathOperator{\opintersect}{\mbox{\sf INTERSECT}}

\newcommand{\kldiv}{D_{\rm KL}}
\newcommand{\pset}[1]{2^{#1}}
\newcommand{\jqed}{\ifnum\journal=1\qed\fi}

\ifnum\journal=0
\newtheorem{definition}{Definition}
\newtheorem{lemma}[definition]{Lemma}
\newtheorem{fact}[definition]{Fact}
\newtheorem{theorem}[definition]{Theorem}

\else
\spdefaulttheorem{fact}{Fact}{\bfseries}{\itshape}
\newcommand{\qedhere}{}
\fi

\newcommand{\knote}[1]{\textcolor{blue}{$\ll$\textsf{#1---Krzysztof}$\gg$\marginpar{$\textcolor{blue}{\Leftarrow}$\iffalse\tiny\bf KO\fi}}}
\newcommand{\vnote}[1]{\textcolor{red}{$\ll$\textsf{#1---Venkat}$\gg$\marginpar{$\textcolor{red}{\Leftarrow}$\iffalse\tiny\bf VG\fi}}}

\newcommand{\descitem}[1]{\item[\quad \sc #1:]}

\begin{document}


\title{Superlinear lower bounds for multipass graph processing\thanks{An extended abstract appeared in the proceedings of the 28th IEEE Conference on Computational Complexity \cite{conf}. Supported in part by NSF CCF-0963975 and the Simons Postdoctoral Fellowship.}}

\ifnum\journal=1

\author{Venkatesan Guruswami \and Krzysztof Onak}

\institute{Venkatesan Guruswami \at
Carnegie Mellon University,
Computer Science Department\\
Gates-Hillman Complex 7211, 
5000 Forbes Avenue, Pittsburgh PA 15213 \\
Tel.: +1-412-268-3041\\
Fax: +1-412-268-5576\\
\email{guruswami@cmu.edu}
\and
Krzysztof Onak \at
IBM T.J.\ Watson Research Center, P.O.\ Box 218, Yorktown Heights, NY 10598\\
Tel.: +1-914-945-1207\\
Fax: +1-914-945-3434\\
\email{krzysztof@onak.pl}
}

\else

\author{{\sf Venkatesan Guruswami}%
\thanks{Supported in part by NSF CCF-0963975.}\\
{\sl Carnegie Mellon University}
\and {\sf Krzysztof Onak}\thanks{Supported by the Simons Postdoctoral Fellowship.}\\{\sl IBM Research}}
\date{February 2016}

\fi

\maketitle

\begin{abstract}
We prove $n^{1+\Omega(1/p)}/p^{O(1)}$ lower bounds for the space complexity of $p$-pass streaming algorithms solving the following problems on $n$-vertex graphs:
\begin{itemize}
\item testing if an undirected graph has a perfect matching (this implies 
lower bounds for computing a maximum matching or even just the maximum matching size),
\item testing if two specific vertices are at distance at most $2(p+1)$ in an undirected graph,
\item testing if there is a directed path from $s$ to $t$ for two specific vertices $s$ and $t$ in a directed graph.
\end{itemize}
The lower bounds hold for $p = O(\log n / \log\log n)$.
Prior to our result, it was known that these problems require $\Omega(n^2)$ space in one pass, but no $n^{1+\Omega(1)}$ lower bound was known for any $p\ge 2$.

These streaming results follow from a communication complexity lower
bound for a communication game in which the players hold two graphs on
the same set of vertices. The task of the players is to find out
whether the sets of vertices at distance exactly $p+1$ from a specific vertex intersect. The game requires a significant amount
of communication only if the players are forced to speak in a specific
difficult order. This is reminiscent of lower bounds for communication
problems such as indexing and pointer chasing. Among other things, our
line of attack requires proving an information cost lower bound for a
decision version of the classic pointer chasing problem and a direct
sum type theorem for the disjunction of several instances of this
problem.
\ifnum\journal=1
\keywords{streaming \and maximum matching \and shortest path \and communication complexity \and information theory}
\fi
\end{abstract}

\section{Introduction}

Graph problems in the streaming model have attracted a fair
amount of attention over the last 15 years. Formally, a streaming
algorithm is presented with a sequence of graph edges and it can read
them one by one in the order in which they appear in the sequence. The
main computational resource studied for this kind of algorithm is the
amount of space it can use, i.e., the amount of information about the
graph the algorithm remembers during its execution.

Due to advances in the storage technology, it is feasible nowadays to collect large amounts of data. Companies store more and more information for reasons that include data mining applications and legal obligations. Sequential access often maximizes readout efficiency in the case of external storage devices. Whenever a single-pass streaming algorithm requires an infeasible amount of main memory, it is natural to ask whether there exists a significantly more efficient algorithm that uses a very small number of passes over data.

At a more theoretical level, relations between nodes (i.e., how they
are connected in the graph and what the distances between them are)
are a fundamental property of graphs that is worth studying. When it
comes to exploring the structure of graphs, allowing for multiple
passes seems to greatly improve the capabilities of streaming
algorithms. For instance, the algorithm of Das Sarma, Gollapudi, and
Panigrahy~\cite{SGP11}, which received the best paper award at PODS 2008, uses
multiple passes to construct long random walks in the graph in order
to approximate PageRank for vertices.  Also many strong lower bounds
of $\Omega(n^2)$ space for one pass easily break if more than one pass
is allowed. This is for instance the case for the early lower bounds
of Henzinger, Raghavan, Rajogopalan~\cite{HRR} and also the lower
bounds of Feigenbaum et al.~\cite{FKMSZ05}.

On the other hand, constructing lower bounds for graph problems is usually based on constructing obstacles for local exploration, and our paper is not different in this respect. We show that finding out if two vertices are at a specific small distance $p$ essentially requires $p/2$ passes to be accomplished in space $O(n)$. The main idea is similar to what is done for pointer chasing. Namely, we place edges in the order opposite to the sequence which enables easy exploration.

\medskip\noindent {\bf Our results.} Let $n$ be the number of vertices in the graph and let $p$ be the allowed number of passes. We show strongly superlinear lower bounds of $T_{n,p}=\Omega \left(\frac{n^{1+1/(2p+2)}}{p^{20} \log^{3/2}n}\right)$ bits of space for three problems:
\begin{itemize}
\vspace{-1ex}
\itemsep=0ex
\item testing if the graph has a perfect matching,
\item testing if two prespecified vertices $u$ and $v$ are at distance at most $2(p+1)$ for an undirected input graph,
\item testing if there is a directed path from $u$ to $v$, where $u$ and $v$ are prespecified vertices and the input graph is directed.
\end{itemize}
In general, lower bounds stronger than $\Omega(n)$ require embedding a
difficult instance of a problem into the ``space of edges'' as opposed
to the ``space of vertices,'' which turns out to be difficult in many
cases. For instance, the $\Omega(n^2)$ lower bounds of~\cite{HRR}
and~\cite{FKMSZ05} do not hold for algorithms that are allowed more
than one pass. 

Communication complexity is a standard tool for proving streaming lower bounds. We describe our hard communication problem from which we reduce to the streaming problems in
Section~\ref{sec:pf-overview}. We now overview related work.

\medskip\noindent {\bf Matchings.} In the maximum matching problem, the goal is to produce a maximum-cardinality set of non-adjacent edges. Streaming algorithms for this problem and its weighted version have received a lot of attention in recent years~\cite{FKMSZ05,McGregor05,EggertKMS12,AhnG11,KonradMM12,GoelKK12,Kapralov13,EpsteinLMS11,Zelke12}.

Our result compares most directly to the lower bound of Feigenbaum et al.~\cite{FKMSZ05}, who show that even checking if a given matching is of maximum size requires $\Omega(n^2)$ space in one pass. Our result can be seen as an extension of their lower bound to the case when multiple passes are allowed. Even when $p\ge 2$ passes are allowed, we show that still a superlinear amount of space, roughly $n^{1+\Omega(1/p)}$, is required to find out if there is a perfect matching in the graph. This of course implies that tasks such as computing a maximum matching or even simply the size of the maximum matching also require this amount of space.

For the approximate version of the maximum matching problem, McGregor~\cite{McGregor05} showed that a $(1-\eps)$-approximation can be computed in $\tilde O(n)$ space\footnote{We use the $\tilde O$ notation to suppress logarithmic factors. For instance, we write $\tilde O(f(n))$ to denote $O\left(f(n) \log^{O(1)}(f(n))\right)$.} with the number of passes that is a function of only $\eps$. The only known superlinear lower bound for the approximate matching size applies only to one-pass algorithms and shows that the required amount of space is $n^{1+\Omega(1 / \log \log n )}$ if a constant approximation factor better than $1-e^{-1}$ is desired \cite{GoelKK12,Kapralov13}.

\medskip \noindent {\bf Shortest paths.} We now move to the problem of computing distances between vertices in an undirected graph. Feigenbaum et al.~\cite{FKMSZ08} show that $\tilde O(n)$ space and one pass suffice to compute an $O(\log n/\log\log n)$-spanner and therefore approximate all distances up to a factor of $O(\log n/\log\log n)$. They also show a closely matching lower bound of $\Omega(n^{1+1/t})$ for computing a factor $t$ approximation to distances between all pairs of vertices.

In the result most closely related to ours, they show that computing the {\em set}\ of vertices at distance $p$ from a prespecified vertex in less than $p/2$ passes requires $n^{1+\Omega(1/p)}$ space. One can improve their lower bound to show that it holds even when the number of allowed passes is $p-1$.
\footnote{This follows by replacing one of their proof components with a stronger pointer chasing result from \cite{GuhaM09}.}
As a result, to compute the distance $p$ neighborhood in $O(n)$ space, essentially the best thing one can do is to simulate the BFS exploration with one step per pass over the input, which requires $p$ passes. In this paper, we show a similar lower bound for the problem of just checking if two specific vertices are at distance exactly $p$. Our problem is algorithmically easier. If two vertices are at distance $p$, $\lceil p/2\rceil$ passes and $O(n)$ space suffice, because one can simulate the BFS algorithm up to the radius of $\lceil p/2 \rceil$ from both vertices of interest. This is one of the reasons why our result cannot be shown directly by applying their lower bound.

A space lower bound of $\Omega(n^2)$ for one pass algorithms to find whether a pair of nodes is at distance $3$ can be found in \cite{FKMSZ05}.

\medskip\noindent {\bf Directed connectivity.} Feigenbaum et al.~\cite{FKMSZ05} show that the directed $u$-$v$ connectivity problem requires $\Omega(n^2)$ bits of space to solve in one pass. However, their lower bound does not extend to more than one pass. Once again our lower bound extends their result to show that a superlinear lower bound holds for multiple passes. (Note that for {\em undirected} graphs, the problem of connectivity can easily be solved with one pass and $\tilde O(n)$ space, using for instance the well known union-find algorithm.)

\medskip\noindent {\bf Optimality of our lower bound.} We conjecture that our lower bounds can be improved from $\tilde\Omega(n^{1+1/(2p+2)})$ to $\tilde\Omega(n^{1+1/(p+1)})$ for $p=O(1)$. For the matching problem, our lower bound is based on showing that finding a single augmenting path is difficult. It is an interesting question if a stronger lower bound can be proved in the case where more augmenting paths have to be found. Currently, no $o(n^2)$-space streaming algorithm is known for this problem with a small number of passes.

\medskip\noindent{\bf Paper organization.} We begin in Section \ref{sec:pf-overview} with a description of the communication problems we study and a high-level overview of our lower bound approach. We set up some useful information-theoretic preliminaries in Section \ref{sec:prelims}. We state our main communication complexity lower bound (Theorem \ref{thm:main}) and use it to show our streaming lower bounds in Section \ref{sec:reductions}. Our communication lower bound is proved in three steps, and we go into the details of these steps in the next three sections. Finally, in Section \ref{sec:main} we put the steps together to give a proof of Theorem \ref{thm:main}.

\section{Proof overview and techniques}\label{sec:pf-overview}

Via simple reductions, our multipass streaming lower bounds for
matching and connectivity reduce to proving communication complexity
lower bounds for a certain decision version of ``set pointer chasing.'' The reductions to streaming are described in Section \ref{sec:reductions}. In the current section we give an  overview of our communication complexity results.  We start with a quick review of information and communication complexity and then introduce communication problems that are useful in our proof.

We assume private randomness in all communication problems, unless otherwise stated. Furthermore, all messages are public, i.e., can be seen by all players (the setting sometimes described as the blackboard model).

\subsection{Information and Communication Complexity}

Communication complexity and information complexity play an important role in our proofs. We now provide necessary definitions for completeness.

The \emph{communication complexity} of a protocol is the function from the input size to the maximum length of messages generated by the protocol on an input of a specific size.
For a problem $\mathcal X$ and $\delta \in [0,1]$, the communication complexity of $\mathcal X$ with error $\delta$ is the function from the input size to the infimum communication complexity of private-randomness protocols that err with probability at most $\delta$ on any input. We write $R_\delta(\mathcal X)$ to denote this quantity.

The \emph{information cost}\footnote{Note that this is the {\em external} information cost following the terminology of \cite{BBCR10}. For product distributions $\psi$, this also equals the {\em internal} information cost. As product distributions will be our exclusive focus in this paper, this distinction is not relevant to us, and we will simply use the term information cost.} $\icost_\psi(\Pi)$ of a protocol $\Pi$ on input distribution $\psi$ equals the mutual information $I(X;\Pi(X))$, where $X$ is the input distributed according to $\psi$ and $\Pi(X)$ is the transcript of $\Pi$ on input $X$.

The \emph{information complexity} $\IC_{\psi,\delta}$ of a problem $\mathcal X$ on a distribution $\psi$ with error $\delta$ is the infimum $\icost_\psi(\Pi)$ taken over all private-randomness protocols $\Pi$ that err with probability at most $\delta$ for any input.

\subsection{Communication Problems}

Consider a communication problem with $p$ players $P_1$, \ldots, $P_p$. Players speak in $r$ rounds and in each round they speak in order $P_1$ through $P_p$. At the end of the last round, the player $P_p$ has to output the solution. We call any such problem a \emph{$(p,r)$-communication problem}. 

We define $[n]$ as $\{1,\ldots,n\}$. For any set $A$, we write $\pset{A}$ to denote the power set of $A$, i.e., the set of all subsets of $A$. For any function $f:A \to \pset{B}$, we define a mapping $\vv{f}:\pset{A}\to\pset{B}$ such that $\vv{f}(S) = \bigcup_{s\in S} f(s)$.

\medskip \noindent {\bf Pointer and Set Chasing.}
The \emph{pointer chasing} communication problem $\PC_{n,p}$, where $n$ and $p$ are positive integers, is a $(p,p-1)$-communication problem in which the $i$-th player $P_i$ has a function $f_i:[n]\to[n]$ and the goal is to compute $f_1(f_2(\ldots f_p(1) \ldots))$.
The complexity of different versions of this problem was explored thoroughly by a number of works \cite{PS84,DGS87,NisanW93,DJS98,PRV01,CCM08,GuhaM08,GuhaM09}.

The \emph{set chasing} communication problem $\SC_{n,p}$, for given positive integers $n$ and $p$, is a $(p,p-1)$-communication problem in which the $i$-th player $P_i$ has a function $f_i:[n]\to\pset{[n]}$ and the goal is to compute $\vv{f_1}(\vv{f_2}(\ldots \vv{f_p}(\{1\}) \ldots))$. A two-player version of the problem was considered by Feigenbaum et al.~\cite{FKMSZ08}.

\medskip\noindent {\bf Operators on Problems.} For a $(p,r)$-communication problem $\mathcal X$,
we write $\opequal(\mathcal X)$ to denote a $(2p,r)$-communication problem in which the first $p$ players $P_1$, \ldots, $P_p$ hold one instance of $\mathcal X$, the next $p$ players $P_{p+1}$, \ldots, $P_{2p}$ hold another instance of $\mathcal X$, and the goal is to output one bit that equals $1$ if and only if the outputs for the instances of $\mathcal X$ are equal. See Figure~\ref{fig:pc_eq} for an example.

\begin{figure*}[th]
\begin{center}	
\includegraphics[scale=.45]{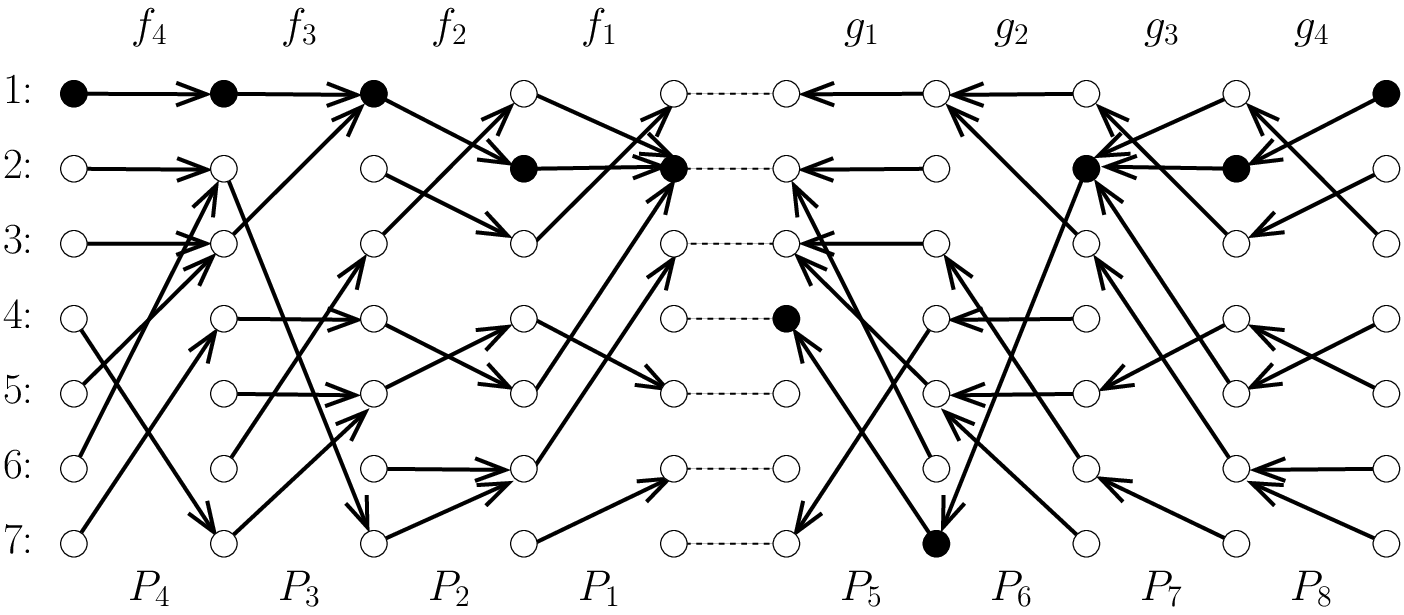}
\end{center}
\caption{A sample instance of $\opequal(\PC_{7,4})$ with a negative solution. It consists of two instances of $\PC_{7,4}$ held by two different sets of players.}
\label{fig:pc_eq}
\end{figure*}

Analogously, for a $(p,r)$-communication problem $\mathcal X$ such that the output is a set, we write $\opintersect(\mathcal X)$ to denote the $(2p,r)$-communication problem in which the first $p$ players $P_1$, \ldots, $P_p$ hold one instance of $\mathcal X$, the next $p$ players $P_{p+1}$, \ldots, $P_{2p}$ hold another instance of $\mathcal X$, and the goal is to output one bit that equals $1$ if and only if the sets that are solutions to the instances of $\mathcal X$ intersect. See Figure~\ref{fig:set_intersect} for an example.

\begin{figure*}[th]
\begin{center}
\includegraphics[scale=.45]{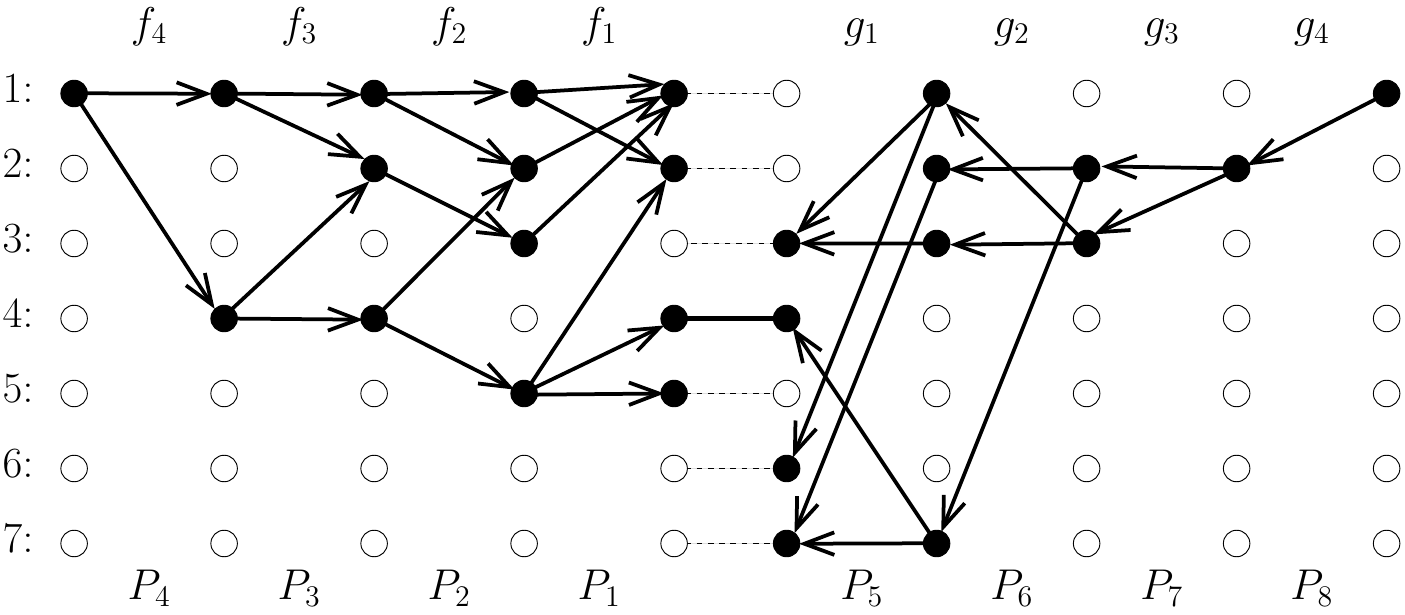}
\end{center}
\caption{A sample instance of $\opintersect(\SC_{7,4})$, where two final sets intersect. The edges outgoing from vertices that are not visited were omitted.}
\label{fig:set_intersect}
\end{figure*}

For a $(p,r)$-communication problem $\mathcal X$ with a Boolean output, we write $\opor_t(\mathcal X)$, where $t$ is a positive integer, to denote the $(p,r)$-communication problem in which players have $t$ instances of $\mathcal X$ and want to output the disjunction of their results.

\medskip\noindent {\bf Limited Pointer Chasing Equality.}
We say that a function $f:A \to B$ is \emph{$t$-colliding}, where $t$ is a positive integer, if there is an $A'\subseteq A$ of size $t$ and a $b \in B$ such that for all $a \in A'$, $f(a) = b$. 

We write ${\LEP_{n,k,t}}$ to denote a modified version of $\opequal(\PC_{n,k})$.
In ${\LEP_{n,k,t}}$ the last player has to output the same value as in $\opequal(\PC_{n,k})$, unless one of the functions in one of the pointer chasing instances is $t$-colliding, in which case the last player has to output $1$.
This is a technical extension to ensure that no element has too many pre-images, which is necessary to make one of our reductions work.

\subsection{Lower bound for $\opintersect(\SC_{n,p})$}
\label{sec:overview}

Our multipass streaming lower bounds for
matching and connectivity reduce to proving a communication complexity
lower bound for the set chasing intersection problem $\opintersect(\SC_{n,p})$.
Note that if the players spoke in the order $P_{2p}$, $P_{2p-1}$, \dots, $P_1$,
then they would be able to solve both instances of $\SC_{n,p}$, using at most $O(n)$ communication per player, which is enough to solve the intersection problem.
If the players spoke in the desired order $P_1$, $P_2$, \dots, $P_{2p}$ but were
allowed a total of $p$ rounds then they would be able to solve the instances of $\SC_{n,p}$
with $O(n)$ communication per message by simulating one step in the pointer chasing instance per round. Our main result is however that if the number of allowed rounds is $p-1$, then approximately $n^{1+\Omega(1/p)}$ bits of
communication are needed to solve the problem, even for randomized
protocols with constant error.

Our result is reminiscent of the classic $\Omega(n)$
communication complexity lower bounds for problems such as indexing
and pointer chasing $\PC_{n,p}$~\cite{NisanW93,GuhaM09} when the players speak in the ``wrong'' order. Guha and McGregor~\cite{GuhaM09} adapt the proof of Nisan and Wigderson~\cite{NisanW93} to show that solving $\PC_{n,p}$ (in $p-1$ passes) requires 
$\Omega(n)/p^{O(1)}$ total communication even if the protocol can be randomized and can err with small constant probability. Increasing the number of rounds to $p$ or letting the players speak in the opposite order (even in just one round) would result in a problem easily solvable with messages of length $O(\log n)$.

Even more directly related is the construction of Feigenbaum et al.~\cite{FKMSZ08}, who show that solving $\SC_{n,p}$ requires $n^{1+\Omega(1/p)}$ communication in less then $p$ passes\footnote{In fact, they show this for roughly less than $p/2$ passes, but replacing the lower bound of \cite{NisanW93} with the lower bound of \cite{GuhaM09} and extending some other complexity results to the setting with multiple players yields the improved bound claimed here.}. Their proof follows by using a direct sum theorem of Jain, Radhakrishnan and Sen~\cite{JRS} to show that solving $t\approx n^{\Theta(1/p)}$ instances of $\PC_{n,p}$ requires roughly $t$ times more communication than solving a single instance. Then they show that an efficient protocol for solving $\SC_{n,p}$ would result in an efficient protocol for solving $t$ instances of $\PC_{n,p}$ in parallel. 

Compared to $\SC_{n,p}$, $\opintersect(\SC_{n,p})$ is a decision problem. In particular, there seems to be no reduction allowing one to reconstruct the sets reached in $\opintersect(\SC_{n,p})$. The only thing that we learn after an execution of the protocol is whether these two sets intersect. Therefore, reducing our question to that of~\cite{FKMSZ08} seems unlikely.

Our proof of the above communication complexity lower bound proceeds
in three steps:
\begin{description}
\descitem{Step A} Reduction to proving a communication lower bound for $\opor_t(\opequal(\PC_{n,p}))$.
\descitem{Step B} A direct sum style step lower bounding the
  communication complexity of $\opor_t(\opequal(\PC_{n,p}))$ as roughly
    $t$ times the information complexity of
    $\opequal(\PC_{n,p})$.
\descitem{Step C} An information complexity lower bound for $\opequal(\PC_{n,p})$.
\end{description}
The technical body of the paper actually proves these steps in the
opposite order (Steps A, B, and C are discussed in Sections \ref{sec:step3}, \ref{sec:step2}, and \ref{sec:step1}, respectively). But here we will expand on the steps in the above order. The actual proof works with a variant of $\opequal(\PC_{n,p})$, namely $\LEP_{n,p,r}$, which we defined earlier, in order to deal with
functions $f_i$ that may be highly colliding, and which may break the reduction in Step~A. For simplicity, we ignore this aspect in the overview, but it is worth keeping in mind that this complicates the execution of Step C on the information complexity lower bound. 

\medskip \noindent {\bf Step A: Reduction to proving a lower bound for
  $\opor_t(\opequal(\PC_{n,p}))$.}  Our idea here is to use a
  communication protocol for $\opintersect(\SC_{n,p})$ to give a
    protocol that can answer if at least one of the $t$ instances of
    $\opequal(\PC_{n,p})$ has a Yes answer, where $t = n^{\Theta(1/p)}$. (Recall that in the
    $\opequal(\PC_{n,p})$ problem, the input consists of two instances
    of $\PC_{n,p}$ with functions $\{f_i,g_i : [n] \to [n]\}_{i=1}^p$
    and the goal is to output Yes iff we end up at the same index in
    both instances, i.e., if $f_1(f_2(\ldots f_p(1) \ldots)) =
    g_1(g_2(\ldots g_p(1) \ldots))$.) Given $t$ instances of
    $\opequal(\PC_{n,p})$, for each instance independently, we
    randomly scramble the connections in every layer while preserving
    the answer to $\opequal(\PC_{n,p})$. We then overlay all these
    instances on top of each other to construct an instance of
    $\opintersect(\SC_{n,p})$ (note that each node has exactly $t$
      neighbors in the next layer).  

By construction, given a Yes
      instance of $\opor_t(\opequal(\PC_{n,p}))$, by following the
      mappings from the instance of $\opequal(\PC_{n,p})$ which has a
      Yes answer, we also obtain an element that belongs to the
      intersection of the two resulting sets in
      $\opintersect(\SC_{n,p})$. Since $t = n^{\Theta(1/p)}$, we have $t^{2p} \ll n$, and we argue that the random scramblings ensure that if none of $t$ instances of $\opequal(\PC_{n,p})$
      have a Yes answer, then it is unlikely that the two resulting
      sets in the instance of $\opintersect(\SC_{n,p})$ will
      intersect. This constraint on $t$ is what limits our lower bound to ${\sim}\,n^{1+1/(2p)}$.

\medskip \noindent {\bf Step B: A direct sum style argument.} In this
step, our goal is to argue that the randomized communication
complexity of $\opor_t(\opequal(\PC_{n,p}))$ is asymptotically
$\Omega(t)$ times larger than that of $\opequal(\PC_{n,p})$. This is
reminiscent of direct sum results of the flavor that computing answers
to $t$ instances of a problem requires asymptotically $t$ times the
resources, but here we only have to compute the OR of $t$ instances.
Our approach is to use the information complexity method that has emerged in the last decade as a potent tool to tackle such direct sum like
questions~\cite{CSWY01,BJKS04,JRS}, and more recently in \cite{BBCR10,BravermanR11} and follow-up works. The introductions of \cite{BBCR10,JPY12} contain more detailed information and references on direct sum and direct product theorems in communication complexity.

Our hard distribution will be the
uniform distribution over all inputs. Being a product distribution,
the information complexity will be at least the sum of the mutual
information between the $i$-th input and the transcript, for $1\le i
\le t$. Using the fact that the probability of a Yes answer on a
random instance of $\opequal(\PC_{n,p})$ is very small (at most
$O(1/n)$), we prove that the mutual information between the $i$-th
input and the transcript cannot be much smaller than the information
cost of $\opequal(\PC_{n,p})$ for protocols that err with probability $o(1/n)$ under the uniform distribution.

\medskip \noindent {\bf Step C: Lower bound for information cost of
  $\opequal(\PC_{n,p})$.} This leaves us with the task of lower
bounding the information cost of low error protocols for
$\opequal(\PC_{n,p})$ under the uniform distribution. This is the most
technical of the three steps. We divide this step into two parts.

First we show that if there were a protocol with low information cost
$\IC$ on the uniform distribution, then there would exist a deterministic
protocol that on the uniform distribution would send mostly short messages
and err with at most twice the probability. This is done by adapting
the proof of the message compression result of \cite{JRS} for bounded
round communication protocols. We cannot use their result as such
since in order to limit the increase in error probability to $\gamma$,
the protocol needs to communicate $\Omega(1/\gamma^2)$ bits. This is
prohibitive for us as we need to keep the error probability as small
as $O(1/n)$, and can thus only afford an additive $O(1/n)$
increase. We present a twist to the simulation obtaining a
deterministic protocol with at most twice the original error
probability. The protocol may send a long message with some
small probability $\eps$ and in other cases communicates at most
$O(\IC/\eps^2)$ bits. In our application, we set $\eps$ to be a
polynomial in $1/p$.

The second part is a lower bound for $\opequal(\PC_{n,p})$ against
such ``typically concise'' deterministic protocols.  To prove this, we
show that if the messages in the deterministic protocol are too short,
then with probability at least $1/2$, the protocol will have little
knowledge about whether the solutions to two instances of pointer
chasing are identical and therefore, will still err with probability
$\Omega(1/n)$, which is significant from our point of view.
The proof extends the lower bound
for pointer chasing due to Nisan and Wigderson \cite{NisanW93} and its
adaptation due to Guha and McGregor~\cite{GuhaM09}. We have to
overcome some technical hurdles as we need a lower bound for the
equality checking version and not for the harder problem of computing
the pointer's value. Further, we need to show that a {\em constant}
fraction of the protocol leaves are highly uncertain about their
estimate of the pointers' values, so that they would err
with probability $\Omega(1/n)$ (with $1/n$ being the collision
probability for completely random and independent values).

Summarizing, Step~C can be seen as a modification of techniques of \cite{NisanW93,GuhaM09} to prove a communication lower bound for $\opequal(\PC_{n,p})$ combined with techniques borrowed from \cite{JRS} to imply a lower bound for information complexity. The relationship between information complexity and communication complexity has been a topic of several papers, starting with \cite{CSWY01,JRS} for protocols with few rounds, and more recently \cite{BBCR10,BravermanR11,Braverman12,ChakrabartiKW12,BravermanW12,KLLRX12} for general protocols.

\section{Preliminaries}
\label{sec:prelims}
\noindent {\bf Constant $C_\star$.} Let $C_\star$ be a constant such that the probability that a function $f:[n] \to [n]$ selected uniformly at random is $C_\star (1+\log n)$-colliding is bounded by $1/(2n^2)$.  The existence of $C_\star$ follows from a combination of the Chernoff and union bounds.

\subsection{Useful information-theoretic lemmas}
\label{subsec:info-theory}
Let us first recall a result that says that if a random variable has large entropy, then it behaves almost like the uniform random variable on large sets.
\begin{fact}[{\cite{RazW89}, see also \cite[Lemma 2.10]{NisanW93}}]\label{fact:almost_uniform}
Let $X$ be a random variable on $[n]$ with $H(X) \ge \log n - \delta$. Let $S\subseteq[n]$ and let $\Delta = \sqrt{\frac{4\delta n}{|S|}}$. If $\Delta  \le 1/10$, then
$\Pr[X \in S] \ge \frac{|S|}{n} \left(1 - \Delta\right).$
\end{fact}
\noindent Using the above result, we show that it is hard to guess correctly with probability $1-o(1/n)$ if two independent random variables distributed on $[n]$ collide if they have large entropy.
\begin{lemma}\label{lem:collision}
Let $X$ and $Y$ be two independent random variables distributed on $[n]$ such that both
$H(X)$ and $H(Y)$ are at least $\log n - \delta$, where $\delta = 48^{-2}$. Then
\begin{itemize}
\itemsep=0ex
 \item $\Pr[X=Y] \ge 1/(8n)$, and 
 \item if $n\ge 4$, $\Pr[X\ne Y] \ge 1/4$.
\end{itemize}
\end{lemma}
\begin{proof}
We first prove that there is a set $S_X\subseteq[n]$ such that $|S_X|\ge\frac{3}{4}n$ and for each $x\in S_X$, $\Pr[X=x] \ge 1/(2n)$. Suppose that there is no such set. Then there is a set $T_X$ of size more than $n/4$ in which every element has probability strictly less than $1/(2n)$, and therefore, $\Pr[X \in T_X] < \frac{|T_X|}{2n}$.
Note that $\sqrt{\frac{4\delta n}{|T_X|}} \le 1/12 < 1/10$, which implies that we can apply Fact~\ref{fact:almost_uniform} to $T_X$. We obtain
$\Pr[X \in T_X] \ge \frac{11}{12}\cdot\frac{|T_X|}{n}$, which contradicts the size of $T_X$ and implies that $S_X$ with the desired properties does exist.

Analogously, one can prove that there is a set $S_Y\subseteq[n]$ such that $|S_Y|\ge\frac{3}{4}n$ and for each $y\in S_Y$, $\Pr[Y=y] \ge 1/(2n)$. Note that $|S_X \cap S_Y|\ge n/2$. For each $x \in S_X \cap S_Y$, $\Pr[X=Y=x] \ge 1/(4n^2)$. Hence 
\[\Pr[X=Y] \ge \sum_{x \in S_X \cap S_Y} \Pr[X=Y=x] \ge 1/(8n).\]
To prove the second claim, for $n \ge 4$, observe that for every setting $x$ of $X$, $|S_y \setminus \{x\}|\ge \frac{3n}{4}-1 \ge n/2$, and therefore, the probability that $Y\ne X$ is at least $|S_y \setminus \{x\}| \cdot 1/(2n) \ge 1/4$.\jqed
\end{proof}

The following lemma gives a bound on the entropy of a variable that randomly selects out of two random values based on another $0$-$1$ valued random variable. We present a simple proof suggested by an anonymous reviewer.

\begin{lemma}\label{lem:entropy_bound}
Let $X_0$, $X_1$, and $Y$ be independent discrete random variables, where $X_0$ and $X_1$ are distributed on the same set $\Omega$ and $Y$ is distributed on $\{0,1\}$. Then 
\[H(X_Y) \le 1 + \sum_{i=0}^1 \Pr[Y=i] \cdot H(X_i).\]
\end{lemma}
\begin{proof}
It follows from basic properties of entropy that
\begin{align*}
H(X_Y) &\le H(Y;X_Y) = H(Y) + H(X_Y|Y)\\
&\le 1 + H(X_Y|Y)
= 1 + \sum_{i=0}^1 \Pr[Y=i] \cdot H(X_i).
\end{align*}
\jqed%
\end{proof}

\section{The Main Tool and Its Applications}
\label{sec:reductions}
The main tool in our paper is the following lower bound for the communication complexity of set chasing intersection.
\begin{theorem}\label{thm:main}
For $n$ larger than some positive constant and $p$ such that $1 < p \le \frac{\log n}{\log \log n}$,
\[R_{1/10}(\opintersect(\SC_{n,p})) = 
\Omega \left(\frac{n^{1+1/(2p)}}{p^{16}\cdot\log^{3/2}n}\right).\]
\end{theorem}
We now present relatively straightforward applications of this theorem to three graph problems in the streaming model.
\begin{theorem}
Solving the following problems with probability at least $9/10$ in the streaming model with $p=O\left(\frac{\log n}{\log \log n}\right)$ passes requires at least $\Omega \left(\frac{n^{1+1/(2(p+1))}}{p^{19}\cdot\log^{3/2}n}\right)$ bits of space:
\begin{description}
 \descitem{Problem 1} For two given vertices $u$ and $v$ in an undirected graph, check if the distance between them is at most $2(p+1)$.
 \descitem{Problem 2} For two given vertices $u$ and $v$ in a directed graph, check if there is a directed path from $u$ to $v$.
 \descitem{Problem 3} Test if the input graph has a perfect matching.
\end{description}
\end{theorem}

\begin{proof}
Let us consider the problems one by one. For Problem 1, we turn an instance of $\opintersect(\SC_{k,p+1})$ into a graph on $n = (2p+3) \cdot k$ vertices.
We modify the graph in Figure~\ref{fig:set_intersect} as follows. First, we make all edges undirected.
Second, we merge every pair of middle vertices connected with a horizontal line into a single vertex. Any path between the top leftmost vertex $u$ and the top rightmost vertex $v$ 
is of length at least $2p+2$.
The length of the path is exactly $2p+2$ if and only if it moves to the next layer in each step. Note that this corresponds to the case that the final sets for two instances of $\SC_{k,p+1}$ intersect. We create the input stream by inserting first the edges describing the function held by $P_1$, then by $P_2$ and so on, until $P_{2p+2}$. If there is a streaming algorithm for the problem that uses at most $T$ bits of space, then clearly there is a communication protocol for $\opintersect(\SC_{k,p+1})$ with total communication $(2p+2) \cdot p \cdot T$ and the same error probability as the streaming algorithm. The protocol can be obtained by the players by simulating the streaming algorithms on their parts of the input and communicating its state. This implies that $T =
\Omega \left(\frac{1}{p^2} \cdot \frac{(n/p)^{1+1/(2(p+1))}}{p^{16}\cdot\log^{3/2}n}\right)
= \Omega \left(\frac{n^{1+1/(2(p+1))}}{p^{19}\cdot\log^{3/2}n}\right)$, where we use the fact that $p^{1/(2(p+1))} = O(1)$.

For Problem 2, the reduction is almost the same, with the only difference being that we make all edges directed from left to right and we want to figure out if there is a directed path from the top leftmost vertex to the top rightmost vertex. Such a path exists if and only if the final sets in the instance of $\opintersect(\SC_{k,p+1})$ intersect. 

\begin{figure*}
\begin{center}
\includegraphics[width=.96\hsize]{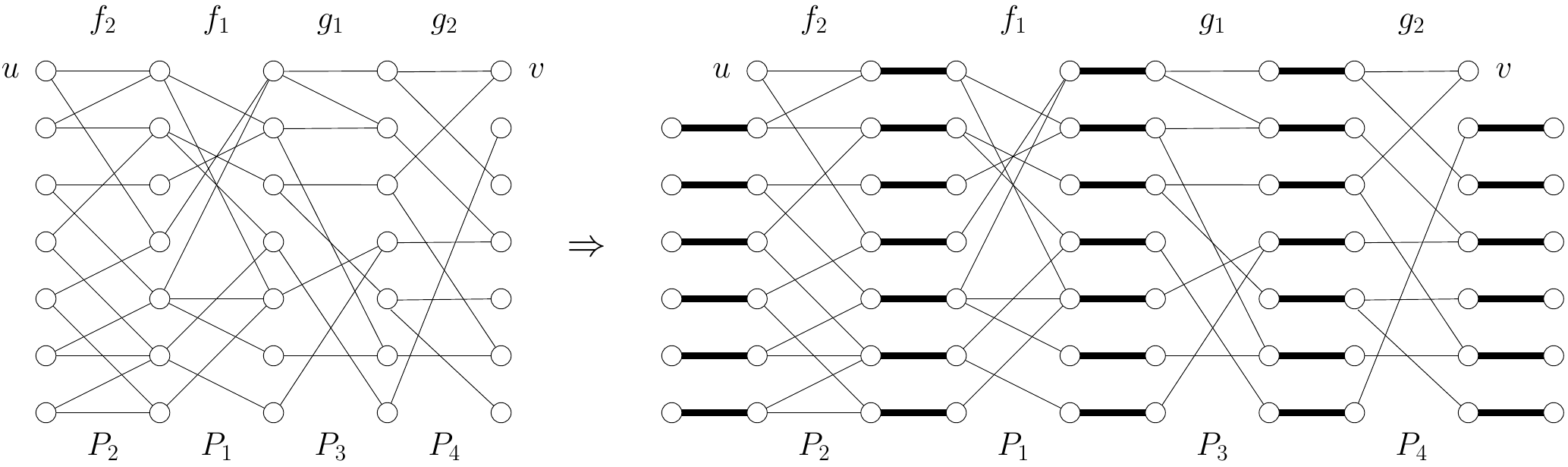}
\end{center}
\caption{Reduction to the perfect matching problem.}
\label{fig:reduction}
\end{figure*}

For Problem~3, the reduction is slightly more complicated. We show how to modify the hard instance $G$ that we have created for Problem~1. Let us first add a perfect matching before and after every layer of edges of the hard instance for Problem~1, except for the first and the last layer, in which we omit one edge. The omitted edges are incident to the vertices $u$ and $v$ corresponding to value 1, i.e., the vertices that we want to connect using a path going directly from left to right in Problem~1. See Figure~\ref{fig:reduction} for an example. Note that the additional edges constitute a matching $M$ in which all but two vertices are matched. Now the graph has a perfect matching iff there exists an augmenting path in $M$ between $u$ and $v$, which are the unmatched vertices. Any augmenting path has to alternate between matched and unmatched edges, which implies in our case, that it has to go directly from left to right. Therefore, any augmenting path in $M$ corresponds to a path going directly from left to right in $G$ and connecting $u$ and $v$. The only difference is that the augmenting path has additional edges coming from the matchings that were inserted into $G$. Therefore the streaming algorithm for testing if a graph has a perfect matching can be used to create a protocol for $\opintersect(\SC_{k,p+1})$, which requires relabeling endpoints of some edges---in order to simulate splitting of vertices---and inserting additional edges at the end of the stream.
\jqed\end{proof}

\section{Step 1: Information Complexity of Pointer Chasing Equality}
\label{sec:ic_lep}
\label{sec:step1}
To prove the main theorem of the paper, we first show a lower bound for the information complexity of Limited Pointer Chasing Equality. 
\begin{lemma}\label{lem:ic_lep}
Let $n$ and $p$ be positive integers such that $n \ge 30p^2$ and $p > 1$.
Then
\begin{align*}
\IC_{\mu,1/(64n)}(\LEP_{n,p,C_\star(1+\log n)})
&\ge \frac{n}{2^{29} \cdot 3^5\cdot p^{16}} - \frac{\lceil2 \log n\rceil}{2^{13}\cdot 3^2 \cdot p^8} - 2\\
&=\Omega\left(\frac{n}{p^{16}}\right)- O\left(\frac{\log n}{p^{8}}+1\right),
\end{align*}
where $\mu$ is the uniform distribution on all possible inputs for the problem.		
\end{lemma}

The proof consists of two smaller steps. First we show that if there is a protocol with low information cost on the uniform distribution, then there is a deterministic protocol that on the uniform distribution sends mostly short messages, and errs with at most twice the probability. Then we show that the messages in the deterministic protocol cannot be too short. Otherwise, with probability at least $1/2$, the protocol would have little knowledge about whether the solutions to two instances of pointer chasing are identical. In this case the protocol would still err with probability $\Omega(1/n)$.

\subsection{Transformation to Deterministic Typically Concise Protocols}

Let us first define concise protocols, which send short messages most of the time.

\begin{definition}
We say that a protocol $P$ is an \emph{$(m,\eps)$-concise} protocol for an input distribution $\mu$ if for each $i$, the probability that the $i$-th message in the protocol is longer then $m$ is bounded by $\eps$.
\end{definition}

The following three facts from \cite{JRS_substate,JRS} are very useful in our proofs. They regard information theory and random variables. For a distribution $P$ on $\mathbb N$, we write $P(i)$ to denote the probability of selecting $i$ from $P$. For two distributions $P$ and $Q$, we write $\kldiv(P\|Q)$ to denote the Kullback-Leibler divergence of $Q$ from $P$.

\begin{fact}[Chain Rule\,{\cite[Fact 1]{JRS}}]\label{fact:JRS_mutual_inf}
Let $X$, $Y$, and $Z$ be random variables. The following identity holds: $I(X;Y,Z) = I(X;Y) + I(X;Z|Y)$.
\end{fact}

\begin{fact}[{\cite[Fact 2]{JRS}}]\label{fact:use_S_for_I}
Let $X$ and $Y$ be a pair of random variables. Let $P$ be the distribution of $Y$ and let $P_x$ be the distribution of $Y$ given $X=x$. Then $I(X;Y) = E_X[\kldiv(P_X\|P)]$.
\end{fact}

\begin{fact}[{\cite[Substate Theorem]{JRS_substate,JRS}}]\label{fact:substate}
Let $P$ and $Q$ be probability distributions on $\mathbb N$ such that $\kldiv(P\|Q) = a$.
Let $\eps \in (0,1)$ and let $\good = \{i \in \mathbb N:P(i)\cdot 2^{-(a+1)/\eps}\le Q(i)\}$.
If $X$ is a random variable distributed according to $P$, then $\Pr[X \in \good] \ge 1-\eps$.
\end{fact}

We now show an auxiliary lemma that shows that if $\kldiv(P\|Q)$ is bounded then a relatively short sequence of independent random variables distributed according to $Q$ suffices to generate a random value from $P$. The lemma is an adaptation of a lemma from \cite{JRS}.

\begin{lemma}\label{lem:good_set}
Let $P$ and $Q$ be two probability distributions on $\mathbb N$ such that $\kldiv(P\|Q) < \infty$. 
Let $(\Gamma_1,\Gamma_2,\Gamma_3,\ldots)$ be a sequence of independent random variables, each distributed according to $Q$.  Let $\Gamma_0 = -1$. Let $\eps \in (0,1)$.
There is a set $\good\subseteq\mathbb N$ and a random variable $R\in\mathbb N$ such that 
\begin{itemize}
\itemsep=1ex
\item $\sum_{i \in \good}P(i) \ge 1-\eps$,
\item for all $x \in \good$, $\Pr[\Gamma_R = x] = P(x)$,
\item $E[R] \le 2^{(\kldiv(P\|Q)+1)/\eps}$.
\end{itemize}
\end{lemma}
\begin{proof}
Let the set $\good$ be defined as in Fact~\ref{fact:substate}, i.e., $\good = \{i \in \mathbb N:P(i) \cdot 2^{-(a+1)/\eps} \le Q(i)\}$, where we set $a = \kldiv(P\|Q)$. Following \cite{JRS}, we use rejection sampling to prove the lemma. Consider the following process. For consecutive positive integers $j$, starting from $1$, do the following. Look at the value $\gamma_j$ taken by $\Gamma_j$. If $\gamma_j \in \good$, toss a biased coin and with probability $P(j) \cdot 2^{-(a+1)/\eps} / Q(j)$, set $R = j$ and finish the process. If $\gamma_j \not\in \good$ or the coin toss did not terminate the process, toss another biased coin and with probability \[2^{-(a+1)/\eps} \cdot \frac{(1 - \sum_{i \in \good}P(i))}{(1 - 2^{-(a+1)/\eps} \cdot \sum_{i \in \good}P(i))},\] set $R=0$ and also terminate the process. Otherwise, continue with $j$ increased by 1. The process terminates with probability $1$.

Let us argue that $R$ and $\good$ have the desired three properties. The first property is a consequence of Fact~\ref{fact:substate}. To prove the other two, observe what happens when the process reaches a specific $j$. The process terminates with $R = j$ and $\Gamma_R = x$ for a specific $x \in \good$ with probability $P(x) \cdot 2^{-(a+1)/\eps}$. The probability that it terminates with $R = 0$ equals
exactly $2^{-(a+1)/\eps}\cdot(1-\sum_{i \in \good}P(i))$. Since these probabilities are independent of $j$, when the process eventually terminates, the probability of $\Gamma_R = x$ for each $x \in \good$ is exactly $P(x)$, which proves the second property. Finally, the probability that the process terminates for a specific $j$ after reaching it is exactly $2^{-(a+1)/\eps}$. Clearly, $E[R]$ is bounded from above by the expected $j$ for which the process stops, which in turn equals exactly $2^{(a+1)/\eps}$. 
\jqed\end{proof}

The following lemma allows for converting protocols with bounded information cost on a specific distribution into deterministic protocols that mostly send short messages on the same distribution. The proof of the lemma is a modification of the message compression result of \cite{JRS}. An important feature of our version is that the error probability is only doubled, instead of an additive constant increase which we cannot afford. A simple but key concept we use to achieve this is to allow the protocol to send long messages with some small (constant) probability. We then handle such ``typically concise'' protocols in our lower bound of Section \ref{subsec:NW-style}.

\begin{lemma}\label{lem:make_concise}
Let $\Pi$ be a private-randomness protocol for a $(p,r)$-communication problem $\mathcal P$ such that $\Pi$ errs with probability at most $\delta>0$ on a distribution $\mu$. For any $q > 0$, there is a protocol $\Pi'$ for $\mathcal P$ such that
\begin{itemize}
 \item $\Pi'$ is deterministic,
 \item $\Pi'$ errs with probability at most $2\delta$ on $\mu$,
 \item $\Pi'$ is $(m,q)$-concise, where $m = 128\cdot(\icost_\mu(\Pi)+2)\cdot(pr/q)^2$.
\end{itemize}
\end{lemma}
\begin{proof}
There are $pr-1$ messages sent in $\Pi$.
We construct a series of intermediate protocols $\Pi'_{pr-1}$, $\Pi'_{pr-2}$, \ldots, $\Pi'_1$, where $\Pi'_{i}$ is a modification of the protocol $\Pi$ in which, for $i \le j \le pr-1$, the $j$-th message is likely to be short. The first $i-1$ messages of $\Pi'_i$ are the same as the messages of $\Pi$. In particular, $\Pi'_i$ uses only private randomness to generate the first $i-1$ messages. Later messages are generated using public randomness. The players in the modified protocols will reveal as much about their inputs as in the original protocol $\Pi$ and therefore, the protocols will err with the same probability, with the only difference being a different encoding of messages and the use of public randomness.

For convenience, let $\Pi'_{pr}$ be the original protocol $\Pi$. We now explain how we convert $\Pi'_{i+1}$ into $\Pi'_i$. Let $M_0$ be the random variable corresponding to the sequence of the first $i-1$ messages in $\Pi'_{i+1}$. Let $M_1$ be the random variable describing the $i$-th message in $\Pi'_{i+1}$. Let $M = (M_0,M_1)$. Recall that $M$ is distributed in the same way as its equivalent for the original protocol $\Pi$.
Let $P_j$ be the player sending the $i$-th message (i.e., $j \equiv i \pmod{p}$).
Let $X$ be the combined inputs of the other players, and let $Y$ be the input of the $j$-th player.
We write $M_1^{m_0}$ to describe the distribution of $M_1$ when $M_0 = m_0$.
Moreover, we write $M_1^{x,y,m_0}$ to describe the distribution of $M_1$ when $X=x$, $Y=y$, and $M_0 = m_0$. The distribution $M_1^{x,y,m_0}$ does not depend on $x$, because the protocol uses only private randomness and to generate the $i$-th message it only uses 
the previous messages and $y$, the input of the $j$-th player. 
It follows from Fact~\ref{fact:JRS_mutual_inf} and Fact~\ref{fact:use_S_for_I} that 
\begin{align*}
I(X,Y;M)
&=I(X,Y;M_0)+I(X,Y;M_1|M_0)\\
&=I(X,Y;M_0)+E_{M_0,X,Y}\left[
\kldiv\left(M_1^{X,Y,M_0}\middle\|M_1^{M_0}\right)\right].
\end{align*}
We define $a_i$ as $E_{M_0,X,Y}\left[
\kldiv\left(M_1^{X,Y,M_0}\middle\|M_1^{M_0}\right)\right]$ for this specific setting of $i$. 
Overall, it follows by induction that the mutual information between the input and the protocol transcript, i.e., $\icost_\mu(\Pi)$,  equals $\sum_{i=1}^{pr-1} a_i$. This also implies that $\kldiv\left(M_1^{x,y,m_0}\middle\|M_1^{m_0}\right) < \infty$ for any setting  $X=x$, $Y=y$, and $M_0=m_0$ that has nonzero probability.

Recall that the first $i-1$ messages of $\Pi'_i$ are generated in the same way as in $\Pi'_{i+1}$. We now describe how $P_j$ generates the $i$-th message. Let $m_0$ be the messages sent so far. The distribution of the $i$-th message, $M_1^{m_0}$ is known to all the players. Let $(\Gamma_1, \Gamma_2, \Gamma_3, \ldots)$ be an infinite sequence of independent random variables, where each $\Gamma_i$ is drawn independently from the distribution $M_1^{m_0}$. The sequence of $\Gamma_i$'s is generated using public randomness, so it is known to all the players as well. 
We now use Lemma~\ref{lem:good_set}, where we set $Q = M_1^{m_0}$, $P = M_1^{x,y,m_0}$, and $\eps = q/8pr$. 
Recall that the distribution $M_1^{x,y,m_0}$ does not depend on $x$, because the randomness is private in the first $i$ messages in $\Pi'_{i+1}$. $P_j$ will reveal as much about its input in $\Pi'_i$ as in $\Pi'_{i+1}$. 
The player $P_j$ fixes a set $\good$ and a random variable $R$ as in Lemma~\ref{lem:good_set}. If $\Gamma_R \in \good$, then the player sends a single bit $0$ followed by a prefix-free encoding of the value $R$. Due to the concavity of the logarithm, the expected length of the message can be bounded by $1 + 16pr(\kldiv(M_1^{x,y,m_0} \| M_1^{m_0})+1)/q + 2 \le 
16pr(\kldiv(M_1^{x,y,m_0} \| M_1^{m_0})+2)/q$, where the additional factor of 2 and additive term of $1$ come from a prefix free encoding.\footnote{This bound can be achieved by unary encoding the length of the message before sending it. If the length of the message is $k$, we first send $k$ zeros proceeded by a single one, which unambiguously identifies the length of the message.}
Overall, the expected length of the message starting with $0$ equals $16pr(a_i+2)/q$.

If $\Gamma_R \not\in \good$, the player generates the message from the part of distribution $M_1^{x,y,m_0}$ restricted to $\mathbb N \setminus \good$ and transmits the selected value prefixing it with a single bit $1$. Overall, all players can decode a message generated according to $M_1^{x,y,m_0}$ and then behave in the same way as in the protocol $\Pi'_{i+1}$.

After applying a sequence of $pr-1$ steps of the transformation, we obtain a randomized protocol $\Pi'_1$ that still errs with probability at most $\delta$.
We now show that there is a suitable setting of random bits in the protocol to obtain the desired deterministic protocol $\Pi'$. In the following, we write $R_\star$ to denote the sequence of random bits used by the protocol. $R_\star$ is a random variable selected from the uniform distribution on infinite binary sequences, which we refer to as $\mathcal R$.
Let $\zeta(R_\star)$ be the probability that $\Pi'_1$ errs on random input from $\mu$ when the internal random bits are set to $R_\star$. We have $E_{R_\star \leftarrow \mathcal R}[\zeta(R_\star)] \le \delta$.
It follows from Markov's inequality that
\[\Pr_{R_\star \leftarrow \mathcal R}\left[\zeta(R_\star) > 2\delta\right] \le 1/2,\]
i.e., the probability that fixing the random bits makes the protocol err with probability higher than $2\delta$ on $\mu$ is at most $1/2$. Consider now the $i$-the message in $\Pi'_1$, where $1 \le i \le pr-1$.
Let $\zeta'_i(R_\star)$ be the probability that the $i$-th message starts with $1$ for the protocol's random bits fixed to $R_\star$. It follows from our construction that 
$E_{R_\star \leftarrow \mathcal R}\left[\zeta'_i(R_\star)\right] \le \eps = q/8pr$.
Applying Markov's inequality, we obtain that
\[\Pr_{R_\star \leftarrow \mathcal R}\left[\zeta'_i(R_\star) > q/2\right] \le 1/4pr,\]
i.e., if we fix the random bits of the protocol, the probability that the $i$-th message starts with $1$ with probability higher than $q/2$ is bounded by $1/4pr$. Finally, let $\zeta''_i(R_\star)$ be the probability that the $i$-th message in $\Pi'_1$ starts with $0$ and has length greater than $128(a_i+2)(pr/q)^2$, given that the random bits of the protocol are set to $R_\star$. Consider a random variable $W_i$ that equals the length of the $i$-th message if the message starts with 0 and 0 if it starts with 1. We know that $E[W_i] \le 16pr(a_i+2)/q$, and therefore, by Markov's inequality, $E_{R_\star \leftarrow \mathcal R}[\zeta''_i(R_\star)] \le q/8pr$.
Applying Markov's inequality again, we obtain that
\[\Pr_{R_\star \leftarrow \mathcal R}\left[\zeta'_i(R_\star) > q/2\right] \le 1/4pr.\]
Summarizing, by fixing the protocol's random bits, with probability at least $1 - 1/2 - 1/(4pr) \cdot (pr - 1) - 1/(4pr) \cdot (pr - 1) = 1/2 - (pr-1)/(2pr) > 1/2 - 1/2 = 0$, we obtain a deterministic protocol that errs with probability at most $2\delta$, and whose $i$-th message, for any $i \in \{1,2,\ldots,pr-1\}$, is longer than $128(a_i+2)(pr/q)^2$ with probability bounded by $q$. The final claim follows from the fact that all $a_i$ are bounded by $\icost_\mu(\Pi)$.
\jqed\end{proof}

\subsection{Lower Bound for Deterministic Typically Concise Protocols}
\label{subsec:NW-style}
In this section, we show that a deterministic concise protocol for the Limited Pointer Chasing Equality cannot send short messages very often, unless it errs with probability $\Omega(1/n)$. The proof follows along the lines of the lower bound for Pointer Chasing due to Nisan and Wigderson \cite{NisanW93} and its adaptation due to Guha and McGregor~\cite{GuhaM09}. The main technical differences come from the fact that we want to show a lower bound for Limited Pointer Chasing Equality. First, this requires ruling out the impact of the easy case when one of the functions is $t$-colliding for large $t$. Second, this requires showing that with constant probability, the last player is unlikely to know what the solutions to the input instances are, and since they are independent, they will collide with probability $\Omega(1/n)$. 

\begin{lemma}\label{lem:concise_lep_lb}
If $n^2 \ge 30 p^2$, then any deterministic $(m,q)$-concise protocol for $\LEP_{n,p,C_\star (1+\log n)}$, where $p > 1$, $m \le \eps n / (4p^2) - \lceil2\log n\rceil$, $q=\frac{1}{12p^2}$, and $\eps = (48p^2)^{-3}$, errs with probability at least $1/(16n)$ on the uniform distribution over all possible inputs.
\end{lemma}

\begin{proof}
Recall that in the $\LEP_{n,p,C_\star (1+\log n)}$ problem, there are $2p$ players $P_1$, \ldots, $P_{2p}$, with players $P_i$ and $P_{p+i}$, $1 \le i\le p$, holding functions $f_i:[n]\to[n]$ and $g_i:[n]\to[n]$, respectively. The goal of the problem is to output ``1'' if one of the functions is $C_\star (1+\log n)$-colliding or $f_1(f_2(\ldots f_p(1)\ldots)) = g_1(g_2(\ldots g_p(1)\ldots))$. Otherwise, ``0'' is the correct output. The players speak in order $P_1$ through $P_{2p}$ and this repeats $p-1$ times.

Let $a_{p+1} = 1$ and by induction, let $a_i = f_i(a_{i+1})$ for each $i\in[p]$.
Analogously, let $b_{p+1} = 1$ and let $b_i = g_i(b_{i+1})$ for each $i\in[p]$.
Unless one of the functions is $C_\star (1+\log n)$-colliding, the goal of the problem is
to determine whether $a_1 = b_1$.

We make two modifications to the protocol:
\begin{enumerate}
\item We augment the $(m,q)$-concise protocol by simulating in parallel the following natural protocol. Initially, we append the pair $(a_{p+1},b_{p+1})$ to each message until we reach the player $P_{p}$, who can compute $a_p=f_p(a_{p+1})$. Then we pass the pair $(a_{p},b_{p+1})$ until it reaches $P_{2p}$, who can compute $b_p = g_p(b_{p+1})$ and pass $(a_p,b_p)$ to the next player.
In general, whenever a message $(a_i,b_i)$ reaches $P_{i-1}$, it is updated to $(a_{i-1},b_i)$, and whenever a message $(a_{i-1},b_i)$ reaches $P_{p+i-1}$, it is updated to $(a_{i-1},b_{i-1})$. This protocol finally computes $(a_2,b_2)$. Appending the information increases the length of each message by $\lceil2\log n\rceil$. This way, we obtain a deterministic $(m+\lceil2\log n\rceil,q)$-concise protocol.

\item
The first time a player whose function is $C_\star(1+\log n)$-colliding is reached in the protocol, we make the player send a message longer than $m+\lceil2\log n\rceil$ bits. This may require modifying other messages sent by the player. We now describe how this can be done depending on the protocol's behavior.

\newcommand{\msg}{\mathfrak{m}}
\begin{enumerate}
\item\label{case:a}
If the player already sends a long message $\msg$ for some input and sequence of previous messages, we make the player send the message $\msg 1$ instead of $\msg$.\footnote{$\msg 1$ denotes here the concatenation of $\msg$ and the message consisting of just a single one.} When the input is $C_\star(1+\log n)$-colliding, we make the player send the message $\msg 0$.
Recall that the player's function is $C_\star(1+\log n)$-colliding with probability at most $1/(2n^2)$. Hence, in this case, the probability of sending a long message increases by at most $1/(2n^2)$.

\item
Likewise, if one of at least $n^2$ prefixes of length $\lceil 2 \log n \rceil$ is not used by the protocol at all, we can use this prefix to transmit long messages. Let $\msg$ be such a prefix. Suppose first that no prefix of $\msg$ can ever be sent the player as a message. In this case, whenever the player's function is $C_\star(1+\log n)$-colliding, we make the player send the message $\msg 0^{m+\lceil2\log n\rceil + 1}$, where $0^{m+\lceil2\log n\rceil + 1}$ denotes a sequence of zeros of length $m+\lceil2\log n\rceil + 1$. If there is a prefix $\msg'$ of $\msg$ that can be sent by the player as a message, we make the player send $\msg'1$ instead of $\msg'$ and we also point out that $\msg'1$ is a short message of length at most $\lceil 2 \log n \rceil$. In this case, we make the player send $\msg' 0^{m+\lceil2\log n\rceil + 1}$ when the function is $C_\star(1+\log n)$-colliding.

In either case, the modifications increase the probability of a long message by at most $1/(2n^2)$ as well.

\item
Finally, if the player does not send long messages and all prefixes of length 
$\lceil 2 \log n \rceil$ are used by the player, at least one of the prefixes occurs with probability at most $1/n^2$. Let $\msg$ be such a prefix. We append $0^{m+\lceil2\log n\rceil + 1}$ to every message that has $\msg$ as a prefix and reduce this case to the first case in which there are long messages. This modification increases the probability of a long message by at most $1/n^2$.

\end{enumerate}

Overall, the probability of long messages can increase by at most $1/(2n^2) + 1/n^2 = \frac{3}{2}n^{-2}$. As a result we obtain a deterministic $(m+\lceil2\log n\rceil,q+\frac{3}{2}n^{-2})$-concise protocol.
\end{enumerate}
Let $m'=m+\lceil2\log n\rceil$ and $q' = q+\frac{3}{2}n^{-2}$.

From now on we think of our deterministic protocol as a decision tree of depth $2p(p-1)$.
The $i$-th layer of nodes, $1 \le i \le 2p(p-1)$, corresponds to the situation when the control is passed to the player $P_j$, where $j\equiv i \pmod{2p}$. Each leaf in the tree is labeled with either a ``0'' or a ``1'', corresponding to the decision made by the algorithm.
Each edge outgoing from nodes at layers 1 through $2p(p-1)-1$ is labeled with the message that the corresponding player sends, given his input and the previous messages. Edges between the last two layers are not labeled, because the last player does not send a message. 

We now introduce a few definitions for each node $z$ in the decision tree:
\begin{itemize}
 \item $c_z$: We set $c_z$ to the total length of the messages sent on the path from the root to $z$.

 \item $F^z_1\times\cdots\times F^z_p \times G^z_1\times\ldots\times G^z_p$:
 Let $\mathcal F$ be the set of all functions from $n$ to $n$. Since the protocol is deterministic, for each node $z$, the set of input functions $(f_1,\ldots,f_p,g_1,\ldots,g_p)$ for which the protocol reaches $z$ can be described as a product $F^z_1\times\cdots\times F^z_p \times G^z_1\times\ldots\times G^z_p \subseteq \mathcal F^{2p}$. Note that if the node is reached then the probability of each tuple in $F^z_1\times\cdots\times F^z_p \times G^z_1\times\ldots\times G^z_p$ is identical. This uses the fact that the initial distribution was uniform.

 \item $i_z$ and $j_z$: We make $i_z$ and $j_z$ be the indices of the last pair $(a_{i_z},b_{j_z})$ sent on the path from the root to $z$. For the root we assume that the pair is $(a_{p+1},b_{p+1}) = (1,1)$, i.e., $i_{\rm root} = j_{\rm root} = p+1$. Recall that for all $z$, $i_z \ge 2$ and $j_z \ge 2$.
 
 \item $(A_z,B_z)$: $(A_z,B_z)$ is a pair of random variables. Its random value is generated by selecting two functions $f_\star\in F^z_{i_z-1}$ and $g_\star \in G^z_{j_z-1}$
 independently and uniformly at random  and applying them to $(a_{i_z},b_{j_z})$ to obtain $(f_\star(a_{i_z}),g_\star(b_{j_z}))$. $(A_z,B_z)$ describes the possible values of the pair $(a_i,b_j)$ if we move one step ahead in applying functions $f_i$ and $g_j$, compared to the trivial algorithm that we simulate in parallel. Since the protocol is deterministic, the inputs of the players are independent, and $A_z$ and $B_z$ depend on inputs of disjoint sets of players, the variables $A_z$ and $B_z$ are independent. 
\end{itemize}

We say that a node $z$ is \emph{confusing} if it has the following properties:
\begin{enumerate}
 \item All messages sent on the path to $z$ have length bounded by $m'$.
 \item $z$ is a leaf or for all $i\in[p]$, both $|F^z_i| \ge 2^{-2c_z} |\mathcal F|$ and $|G^z_i| \ge 2^{-2c_z}|\mathcal F|$, where $c_z$ is the total length of the messages on the path from the root to $z$.
 \item $H(A_z) \ge \log n - \delta$
 and $H(B_z) \ge \log n - \delta$, where $\delta = \eps^{2/3}$.
\end{enumerate}

It is easy to see that the root of the decision tree is confusing. We now prove by induction that the probability that for a random input, the protocol reaches a non-confusing node in step $i$ is bounded by $(i-1)\cdot(q + \frac{5}{2}n^{-2} + 4\eps^{1/3})$. Suppose that the claim is true for step $i$ and let us prove it for step $i+1$. We bound the probability that a specific property is violated.
\begin{enumerate}
\item The probability that the first property is violated
is bounded by $q'$, because the protocol is $(m',q')$-concise. 

\item Consider a confusing node $z$ in step $i$. If $i=2p(p-1)$, the children of $z$ are leaves, and the property holds. So it suffices to focus on the case that $i<2p(p-1)$. What is the probability that the second property is violated for some child $w$ of $z$? Let $P_j$ be the player in control of step $i$. Without loss of generality, let us assume that $1 \le j \le p$. Note that for all $t \in [p]$, $G^w_t = G^z_t$ and for all $t \in [p]\setminus\{j\}$, $F^w_t = F^z_t$. The property may only be violated for $F^w_j$. For each child $w$ of $z$, let $m_w = c_w - c_z$. Let $W$ be the random variable representing the distribution of children of $z$. Then
\begin{align*}
\Pr\left[\frac{|F^W_j|}{|\mathcal F|} < 2^{-2c_W} \right] 
&\le 
\Pr\left[\frac{|F^W_j|}{|F^z_j|} < 2^{-2m_W} \right]
\le \sum_w 2^{-2m_w}\\
& \le \frac{1}{n^2}\sum_w 2^{-m_w} \le \frac{1}{n^2},
\end{align*}
where the second to last inequality follows from the fact that $m_w\ge2\log n$, and the last inequality follows from Kraft's inequality. Therefore the probability that a confusing node loses the second property in the next step is bounded by $1/n^2$.

\item It remains to bound the probability that the third property is lost. Let $z$ be a confusing node in step $i$ and let $P_j$ be the player in charge of this step. If neither $j = i_z - 1$ nor $j-p = j_z - 1$, then for any child $w$ of $z$, $F^w_{i_z-1} = F^z_{i_z-1}$ and $G^w_{j_z-1} = G^z_{j_z-1}$. In this case the pairs of variables $(A_z,B_z)$ and $(A_w,B_w)$ have the same distribution and therefore the respective entropies remain the same. Consider now the case that $j = i_z - 1$. $P_j$ computes $a_j = f_j(a_{i_z})$ and we need to bound the entropy $H(A_w)$ of $A_w$ for all children $w$ of $z$, which is essentially the entropy of $a_{j-1}$ given all the information communicated so far. The information about $f_{j-1}$ at each child $w$ can be expressed as a vector $f^w_{j-1} =(f^w_{j-1}(1),f^w_{j-1}(2),\ldots,f^w_{j-1}(n))$ of random variables in $[n]$. We have $H(A_w) = H(f_{j-1}^w(a_j))$. Moreover,
\begin{align*}
\sum_{t \in [n]} H(f^w_{j-1}(t))&\ge 
H(f^w_{j-1})\\
&= \log|F^w_{j-1}|\\
&= \log|F^z_{j-1}|\\
&\ge \log(2^{-2c_z}|\mathcal F|)\\
&=\log |\mathcal F| - 2c_z \\
&\ge n(\log n - \eps).
\end{align*}
The first inequality above follows from subadditivity of entropy. The second and third inequalities follow from the fact that the function is uniformly distributed on $F^w_{j-1}=F^z_{j-1}$ of size bounded by the fact that $z$ is confusing (Property 2). Finally the last inequality follows from the fact that $z$ is confusing (Property 1), and therefore,
$\eps n / 2 \ge 2p(p-1) \cdot m' \ge c_z$.

For $t$ uniformly distributed on $[n]$, by Markov's inequality, we have
\[\Pr_t[H(f^w_{j-1}(t)) \le \log n - \delta] \le \eps/\delta.\]
Unfortunately, $a_j$ may not be uniformly distributed. However, we can exploit its high entropy, at least $\log n - \delta$. We apply Fact~\ref{fact:almost_uniform}. Let $S$ be the set of $t$ such that
$H(f^w_{j-1}(t)) \ge \log n - \delta$. We already know that $|S| \ge (1- \eps/\delta)n$.
Note that we can apply Fact~\ref{fact:almost_uniform}, because
\[\Delta = \sqrt{\frac{4\delta}{|S|/n}} \le \sqrt{\frac{4\delta}{1-\frac{\eps}{\delta}}}
\le \sqrt{8\delta} = \sqrt{8}\cdot \eps^{1/3}
\le \sqrt{8}\cdot 48^{-1} < 1/10.\]
The probability that $a_j$ belongs to $S$ is at least
\begin{align*}
\frac{|S|}{n}\left(1 - \sqrt{\frac{4\delta}{|S|/n}}\right)
&\ge \left(1-\frac{\eps}{\delta}\right)\left(1 - \sqrt{\frac{4\delta}{1-\frac{\eps}{\delta}}}\right)
\ge \left(1-\eps^{1/3}\right)\left(1-\sqrt{8\delta}\right)\\
&\ge 1 - (1+\sqrt{8})\eps^{1/3} \ge 1-4\eps^{1/3}.
\end{align*}
This implies that 
\[\Pr_{a_j}[H(f^w_{j-1}(a_j)) \le \log n - \delta] \le 4\eps^{1/3}.\]
The case that $j-p = j_z - 1$ is analogous, and therefore, the probability that the third property is lost in the next step is bounded by $4\eps^{1/3}$.
\end{enumerate}
Summarizing, the probability that moving from step $i$ to step $i+1$, we move from a confusing node to a non-confusing node is bounded by $q' + 1/n^2 + 4\eps^{1/3}= q + \frac{5}{2}n^{-2} + 4\eps^{1/3}$, which finishes the inductive proof. 

Overall, it follows that the protocol finishes in a non-confusing leaf with probability bounded by $2p(p-1) \cdot (q + \frac{5}{2}n^{-2} + 4\eps^{1/3}) \le 2p^2 \cdot q + \frac{5p^2}{n^2} + 8p^2 \cdot (48p^2)^{-1} \le 1/6 + 1/6 + 1/6 = 1/2$.

Consider now a confusing leaf $z$. Recall that we modified the protocol so that if one of the functions is $C_\star (1 + \log n)$-colliding a message longer than $m'$ is transmitted. By definition, in such a case, the simulation of the protocol leads to a leaf that is not confusing.
Therefore, the correct solution to an input instance that leads to $z$ is solely based on whether $a_1 = b_1$. We know that the random variables $A_z$ and $B_z$, which model $a_1$ and $b_1$, respectively, are independent and both have entropy at least $\log n - \delta$, where $\delta \le 48^{-2}$. 
Observe also that $n \ge \sqrt{18p^2} \ge 4$. Hence it follows from Lemma~\ref{lem:collision} that whatever solution the protocol claims at $z$, be it ``0'' or ``1'', the claim is incorrect with probability at least $1/(8n)$. Overall, on all inputs the protocol has to err with probability at least $\frac{1}{2}\cdot\frac{1}{8n} = \frac{1}{16n}$.
\jqed\end{proof}

\subsection{Proof of Lemma~\ref{lem:ic_lep}}

We now combine Lemmas~\ref{lem:make_concise} and~\ref{lem:concise_lep_lb} to prove Lemma~\ref{lem:ic_lep}, the main result of Section~\ref{sec:ic_lep}.

\ifnum\journal=0
\begin{proof}[Proof of Lemma~\ref{lem:ic_lep}]
\else
\begin{proof}[Lemma~\ref{lem:ic_lep}]
\fi
Consider any protocol protocol $\Pi$ for $\LEP_{n,p,C_\star(1+\log n)}$ that errs with probability at most $1/(64n)$ on $\mu$. 
By Lemma~\ref{lem:make_concise}, there is a deterministic $(m,1/(12p^2))$-concise protocol $\Pi'$ for $\LEP_{n,p,C_\star(1+\log n)}$ that errs with probability at most $1/(32n)$ on $\mu$,
where $m = 128 \cdot (\icost_\mu(\Pi)+2) \cdot \left(\frac{2p^2}{(1/(12p^2))}\right)^2
= 2^{13}\cdot 3^2 \cdot p^8 \cdot (\icost_\mu(\Pi)+2)$. It follows from Lemma~\ref{lem:concise_lep_lb} that $m \ge \frac{n}{(48p^2)^3\cdot 4p^2} - \lceil2 \log n\rceil = \frac{n}{2^{16} \cdot 3^3\cdot p^8} - \lceil2 \log n\rceil$.
Therefore, $\icost_\mu(\Pi) \ge \frac{n}{2^{29} \cdot 3^5\cdot p^{16}} - \frac{\lceil2 \log n\rceil}{2^{13}\cdot 3^2 \cdot p^8} - 2$. Since this bound holds for any protocol $\Pi$ that is correct with probability $1-1/(64n)$, this is also a lower bound for the information complexity of the problem.
\jqed\end{proof}

\section{Step 2: Direct Sum Theorem for Pointer Chasing Equality}
\label{sec:step2}
The following lemma is the main result of this section.

\begin{lemma}\label{lem:step2}
Let $n$, $p$, and $t$ be integers such that 
$p > 1$, $n^2 \ge 30p^2$, and $t \le n/4$.
Let $r = C_\star(1+\log n)$. Then
\[R_{1/3}(\opor_t(\LEP_{n,p,r})) = \Omega\left(\frac{tn}{p^{16}\log n}\right)- O\left(pt^2\right).\]
\end{lemma}

Before we prove it, let us first recall two classic results. First, the information complexity is a lower bound for the randomized communication complexity of a protocol that errs with the same probability.

\begin{lemma}[{\cite[Proposition 4.3]{BJKS04}}]\label{lem:inf_vs_com}
Let $\delta \in (0,1)$. For any communication problem $\mathcal P$ and any distribution $\psi$ on inputs, 
$R_{\delta}(\mathcal P) \ge \IC_{\psi,\delta}(\mathcal P)$.
\end{lemma}

Second, if the input distribution is a product distribution on multiple instances of a subproblem, then the total information revealed by the protocol transcript equals at least the information revealed for each of the instances.

\begin{lemma}[Follows from {\cite[Lemma 5.1]{BJKS04}}]\label{lem:coordinatewise}
Let $\mathcal P$ be a communication problem with Boolean output and let $\Pi$ be a private-randomness protocol for $\opor_t(\mathcal P)$ for a positive integer $t$. Let $\psi$ be a distribution on inputs of $\mathcal P$. Let $X=(X_1,\ldots,X_t)$ be a vector of independent random variables with each distributed according to $\psi$. For any input $x$, let $\Pi(x)$ be the transcript of $\Pi$ on $x$. Then the following inequality holds: 

\smallskip $\qquad\qquad$ $I(X;\Pi(X)) \ge \sum_{i=1}^{t} I(X_i;\Pi(X))$.
\end{lemma}

Now we show the main ingredient, which is a proof that any correct protocol for 
$\opor_t(\LEP_{n,p,r})$ has to reveal almost as much information about each coordinate as if it was separately solving the corresponding instances of $\LEP_{n,p,r}$.
\begin{lemma}\label{lem:single_coordinate}
Let $\Pi$ be a private-randomness protocol for $\opor_t(\LEP_{n,p,r})$ that errs with probability at most $\delta$. Let $X=(X_1,\ldots,X_t)$ be a random vector with each coordinate $X_i$ independently selected from the uniform distribution $\mu$ on all possible inputs to $\LEP_{n,p,r}$.
Let $\Pi(x)$ be the transcript of $\Pi$ on input $x$.

If $r\ge C_\star \cdot (\log n + 1)$, $t \le n/4$, and $p \le n$, then
for each $i \in [t]$, 
\[I(X_i;\Pi(X)) \ge \IC_{\mu,2\delta}(\LEP_{n,p,r}) - 6pt\log n.\]
\end{lemma}

\begin{proof}
For each $j \in [t]$, let $Y_j\in\{0,1\}$ be the solution to $\LEP_{n,p,r}$ on a specific coordinate $X_j$. The probability that $Y_j=1$ is bounded by $1/n + 2p \cdot 1/(2n^2) \le 2/n$, where the first term comes from the probability that the equality of two instances of $\PC_{n,p}$ holds and the other is a bound on the probability that one of the functions is $r$-colliding.

Fix $i \in [t]$. By the union bound,
\begin{equation}
\Pr[\bigvee_{j\ne i} Y_j = 1] \le (t-1) \cdot \frac{2}{n} < \frac{1}{2}.\label{eq:probV}
\end{equation}
If $\bigvee_{j\ne i} Y_j = 0$, the solution to $\opor_t(\LEP_{n,p,r})$
on the input instance equals $Y_i$. Therefore, $\Pi$ has to compute $Y_i$ with probability at least $1-2\delta$, provided $\bigvee_{j\ne i} Y_j = 0$. If it erred with higher probability, it would overall err on the input instance of $\opor_t(\LEP_{n,p,r})$
with probability greater than $\delta$.

We now bound $I(X_i;\Pi(X)|\bigvee_{j\ne i} Y_j = 0)$ from below, using $\IC_{\mu,2\delta}(\LEP_{n,p,r})$. To achieve this goal, we construct a \emph{private-randomness} protocol $\Pi'$ for $\LEP_{n,p,r}$ that obtains as input a uniformly selected instance $X_i$ of $\LEP_{n,p,r}$, selects $X_j$'s, for $j \ne i$, uniformly from those with solution 0, and emulates $\Pi$ on the resulting full input. For each $j\in[t]$, $X_j$ is a set of functions $f_{j,l}$ and $g_{j,l}$, $l \in [p]$, with each function held by a different player. We write $a_{j,l}$ and $b_{j,l}$, where $j\in[t]$ and $l \in [p+1]$, to refer to intermediate pointer chasing values. Formally, for each $j\in[t]$, $a_{j,p+1} = b_{j,p+1} = 1$, and we recursively define $a_{j,l} = f_{j,l}(a_{j,l+1})$ and $b_{j,l} = g_{j,l}(b_{j,l+1})$ for $l \in [p]$. We want to ensure that for each $j \ne i$, the players obtain their set of functions $f_{j,l}$ and $g_{j,l}$ uniformly from inputs such that the solution to $X_j$ is $0$. It suffices that the players collectively select values $a_{j,l}$ and $b_{j,l}$, $1 \le l \le p$, uniformly at random from all configurations but those with $a_{j,1} = b_{j,1}$, since all of them are equally likely due to symmetry. Then the remaining values of functions can be selected by each player independently, uniformly at random from the set of those that do not result in $r$-colliding functions.
In our protocol $\Pi'$, we make the first player select all $a_{j,l}$'s and $b_{j,l}$'s (overall $2p(t-1)$ values) and send them at the beginning of the very first message. Then the players emulate $\Pi$. Therefore, the transcript $\Pi'(X_i)$ of $\Pi'$ starts with a configuration of $a_{j,l}$'s and $b_{j,l}$'s, which is followed by the transcript of the emulation of $\Pi$. We write $\Pi'_1(X_i)$ and $\Pi'_2(X_i)$ to denote the first and second part of the transcript, respectively. Since $\Pi'$ solves $\LEP_{n,p,r}$ with probability at least $1-2\delta$ on the uniform distribution, $I(X_i;\Pi'(X_i)) \ge \IC_{\mu,2\delta}(\LEP_{n,p,r})$. We have
\begin{align*}
I(X_i;\Pi'(X_i)) &= I(X_i;\Pi'_1(X_i),\Pi'_2(X_i))\\
&= H(\Pi'_1(X_i),\Pi'_2(X_i)) - H(\Pi'_1(X_i),\Pi'_2(X_i)|X_i)\\
&\le H(\Pi'_1(X_i)) + H(\Pi'_2(X_i)) - H(\Pi'_2(X_i)|X_i)\\
&= H(\Pi'_1(X_i)) + I(X_i;\Pi'_2(X_i))\\
&\le 2p(t-1)\log n + I(X_i;\Pi'_2(X_i)),
\end{align*}
where the first inequality follows from the basic properties of entropy and the second from the fact that $\Pi'_2(X_i)$ consists of $2p(t-1)$ values in $[n]$. Since $\Pi'_2(X_i)$ has exactly the same distribution with respect to $X_i$ as $\Pi(X)$ with respect to $X_i$ (under the restriction on $Y_j$ for $j \ne i$), $I(X_i;\Pi(X)|\bigvee_{j\ne i} Y_j = 0) = 
I(X_i;\Pi'_2(X_i)) \ge I(X_i;\Pi'(X_i)) - 2p(t-1)\log n \ge \IC_{\mu,2\delta}(\LEP_{n,p,r}) - 2p(t-1)\log n$.

We now use $I(X_i;\Pi(X)|\bigvee_{j\ne i} Y_j = 0)$ to bound $I(X_i;\Pi(X))$ from below.
By definition, we have $I(X_i;\Pi(X)) = H(X_i) - H(X_i|\Pi(X))$.
Note that $H(X_i) = H(X_i | \bigvee_{j\ne i} Y_j = 0)$, because the coordinates of $X_i$ are independent. Let us now upper bound $H(X_i|\Pi(X))$.
In order to bound $H(X_i | \Pi(X))$, we split the probabilistic space based on the value of $\bigvee_{j\ne i} Y_j$, which is either 0 or 1. We apply Lemma~\ref{lem:entropy_bound} to this partition, which bounds the total entropy using a convex combination of entropies for each of the cases, with an extra additive term of 1.
\begin{align*}
&H(X_i\mid\Pi(X)) = E_\pi[H(X_i\mid\Pi(X) = \pi)]\\
&\le E_\pi  \biggl[ 1+\sum_{z=0}^{1}\Pr\Bigl[\bigvee_{j\ne i} Y_j = z\mid \Pi(X) = \pi\Bigr]\cdot H\Bigl(X_i\mid\Pi(X) = \pi,\bigvee_{j\ne i} Y_j = z\Bigr) \biggr].
\end{align*}
For each transcript $\pi$ and each $z \in \{0,1\}$, let $p_{z,\pi} = \Pr[\Pi(X) = \pi | \bigvee_{j\ne i} Y_j = z]$.
After a few more straightforward transitions, we obtain
\begin{align*}
&H(X_i\mid\Pi(X))\\
&\le 1 + E_\pi\Biggl[\sum_{z=0}^{1}
p_{z,\pi} \cdot \frac{\Pr\bigl[\bigvee_{j\ne i} Y_j = z\bigr]}{\Pr[\Pi(X) = \pi]}
\cdot H\Bigl(X_i\mid \Pi(X) = \pi,\bigvee_{j\ne i} Y_j = z\Bigr) \Biggr]\\
&= 1+ \sum_{z=0}^{1}\Pr\Bigl[\bigvee_{j\ne i} Y_j = z\Bigr] \cdot E_\pi\biggl[\frac{p_{z,\pi}}{\Pr [\Pi(X) = \pi]}
\cdot H\Bigl(X_i\mid\Pi(X) = \pi,\bigvee_{j\ne i} Y_j = z\Bigr) \biggr]\\
&= 1+ \sum_{z=0}^{1}\Pr\bigl[\bigvee_{j\ne i} Y_j = z\bigr] \cdot 
\sum_{\pi} p_{z,\pi} \cdot H\Bigl(X_i\mid\Pi(X) = \pi,\bigvee_{j\ne i} Y_j = z\Bigr)\\
&= 1+ \sum_{z=0}^{1}\Pr\bigl[\bigvee_{j\ne i} Y_j = z\bigr] \cdot 
H\Bigl(X_i \mid \Pi(X), \bigvee_{j\ne i} Y_j = z \Bigr).
\end{align*}
Note that the entropy of $X_i$, and therefore also any conditional entropy of $X_i$, is always bounded by 
$2pn\log n$. Hence
\begin{align*}
H(X_i \mid \Pi(X)) &\le 1 + H\Bigl(X_i\mid\Pi(X),\bigvee_{j\ne i} Y_j = 0\Bigr)
+ \Pr\bigl[\bigvee_{j\ne i} Y_j = 1\bigr] \cdot 2pn\cdot\log n\\
&\le 1 + H\Bigl(X_i\mid\Pi(X),\bigvee_{j\ne i} Y_j = 0\Bigr)
+ \frac{2(t-1)}{n} \cdot 2pn\log n\\
&\le 1 + H\Bigl(X_i\mid\Pi(X),\bigvee_{j\ne i} Y_j = 0\Bigr)
+ 4pt\log n,
\end{align*}
where the second inequality uses Equation~\ref{eq:probV}.
Thus we obtain
\begin{align*}
I(X_i;\Pi(X)) &= H(X_i) - H(X_i \mid \Pi(X))\\
&\ge H\Bigl(X_i\mid\bigvee_{j\ne i} Y_j = 0\Bigr)
- H\Bigl(X_i\mid\Pi(X),\bigvee_{j\ne i} Y_j = 0\Bigr)
- 1 - 4pt\log n\\
&\ge I\Bigl(X_i;\Pi(X)\mid\bigvee_{j\ne i} Y_j = 0\Bigr)
- 1 - 4pt\log n\\
&\ge \IC_{\mu,2\delta}(\LEP_{n,p,r}) - 6pt\log n.\qedhere
\end{align*}%
\jqed\end{proof}

We can finally prove the main lemma of this section.

\ifnum\journal=0
\begin{proof}[Proof of Lemma~\ref{lem:step2}]
\else
\begin{proof}[Lemma~\ref{lem:step2}]
\fi
Let $\Pi$ be a private-randomness protocol for $\opor_t(\LEP_{n,p,r})$ that errs with probability at most $1/(128n)$.
It follows from Lemmas~\ref{lem:coordinatewise} and~\ref{lem:single_coordinate} that
\[I(X;\Pi(X)) \ge \sum_{i=1}^{t}I(X_i;\Pi(X))
\ge t\cdot\IC_{\mu,1/(64n)}(\LEP_{n,p,r}) -6pt^2 \log n.\]
By definition, this quantity bounds also $\IC_{\mu^t,1/(128n)}(\opor_t(\LEP_{n,p,r}))$.
Therefore, by Lemma~\ref{lem:inf_vs_com} and Lemma~\ref{lem:ic_lep}, we get
\[R_{1/(128n)}(\opor_t(\LEP_{n,p,r}))
\ge \Omega\left(\frac{tn}{p^{16}}\right)- O\left(pt^2\log n\right).\]
Via standard amplification bounds, $R_{1/(128n)}(\opor_t(\LEP_{n,p,r})) \le R_{1/3}(\opor_t(\LEP_{n,p,r})) \cdot O(\log n)$, which gives us
\[ R_{1/3}(\opor_t(\LEP_{n,p,r})) = \Omega\left(\frac{tn}{p^{16}\log n}\right)- O\left(pt^2\right).\qedhere\]
\jqed\end{proof}
\section{Step 3: Reduction to Set Chasing Intersection}
\label{sec:step3}
We give a reduction showing that under specific conditions, a protocol for $\opintersect(\SC_{n,p})$ can be used to create a communication protocol for $\opor_t(\LEP_{n,p,r})$.

\begin{lemma}\label{lem:set_inters_vs_or_lep}
Let $n$, $p$, $r$, and $t$ be positive integers such that $t^{2p} r^{p-1}\le n/10$.
If there is a communication protocol for $\opintersect(\SC_{n,p})$ that uses
$C$ bits of communication and errs with probability at most $1/10$, then there is a public-randomness communication protocol for $\opor_t(\LEP_{n,p,r})$ that uses
$C + 2p$ bits of communication and errs with probability at most $2/10$.
\end{lemma}
 
\begin{proof}
Consider an instance of $\opor_t(\LEP_{n,p,r})$. There are $2p$ players, who have $t$ instances of $\LEP_{n,p,r}$. Each instance of $\LEP_{n,p,r}$ consists of two instances of $\PC_{n,p}$. Let $f_{i,j}$ and $g_{i,j}$, where $1 \le i \le p$ and $1 \le j \le t$,  be the functions that describe these two instances. For each $i \in [p]$, player $i$ knows $f_{i,j}$ and player $p+i$ knows $g_{i,j}$.
If any of the functions $f_{i,j}$ or $g_{i,j}$ is $r$-colliding, then the solution to the problem is $1$. The players can check if this is the case in one round of communication with each player communicating only one bit. It therefore suffices to show a protocol
that solves 
$\opor_t(\opequal(\PC_{n,p}))$, i.e., computes
\[\bigvee_{j=1}^{t}\Bigl(
f_{1,j}(f_{2,j}(\ldots f_{p,j}(1) \ldots))
 = 
g_{1,j}(g_{2,j}(\ldots g_{p,j}(1) \ldots))
\bigr),\]
using $C$ bits communication, under the assumption that no $f_{i,j}$ or $g_{i,j}$ is 
$r$-colliding. To this end, we show a randomized reduction of this problem to $\opintersect(\SC_{n,p})$.

First, using common randomness, the players select random permutations $\pi_{i,j},\rho_{i,j}:[n] \to [n]$ for $1 \le i \le p$ and $1 \le j \le t$.
Permutations are selected independently, except that $\pi_{1,j} = \rho_{1,j}$ for all $1 \le j \le t$.
Furthermore they are generated using public randomness, so they are known to all players.
(For functions $\lambda:A \to B$ and $\tau:B \to C$, we write $\tau \circ \lambda$ to denote the function from $A$ to $C$ such that $(\tau \circ \lambda)(x) = \tau(\lambda(x))$ for all $x\in A$.)
For all $1 \le j \le t$, let 
\[f'_{p,j} = \pi_{p,j} \circ f_{p,j} \quad \text{ and } \quad g'_{p,j} = \rho_{p,j} \circ g_{p,j}.\]
For all $1 \le i \le p-1$ and $1 \le j \le t$,
let 
\[f'_{i,j} =  \pi_{i,j} \circ f_{i,j} \circ  \pi^{-1}_{i+1,j}\quad \text{ and } \quad g'_{i,j} =  \rho_{i,j} \circ g_{i,j} \circ\rho^{-1}_{i+1,j}.\]
It is easy to see that an instance of $\opor_t(\opequal(\PC_{n,p}))$ with $f'_{i,j}$ and $g'_{i,j}$ is equivalent to the original instance with $f_{i,j}$ and $g_{i,j}$. The permutations randomly relabel intermediate and final values with final values relabeled in the same way on both sides.

We construct an instance of $\opintersect(\SC_{n,p})$ by giving the $i$-th player, $1 \le i \le p$, a function $f^\star_{i}:[n] \to \pset{[n]}$
such that for any $x \in [n]$,
\[f^\star_{i}(x) = \{f'_{i,j}(x): 1 \le j \le t\}\]
and by giving the $p+i$-th player, $1 \le i \le p$, 
a function $g^\star_{i}:[n] \to \pset{[n]}$ such that for any $x \in [n]$,
\[g^\star_{i}(x) = \{g'_{i,j}(x): 1 \le j \le t\}.\]
The goal in this instance is to compute 
\[\left(\vv{f^\star_1}(\vv{f^\star_2}(\ldots \vv{f^\star_p}(\{1\}) \ldots))
\cap
\vv{g^\star_1}(\vv{g^\star_2}(\ldots \vv{g_p^\star}(\{1\}) \ldots))\right) \ne \emptyset.\]

The instance of $\opintersect(\SC_{n,p})$ that we have just defined can be seen as stacking mappings from different instances of $\opequal(\PC_{n,p})$ on top of each other. Instead of following a single function $f'_{i,j}$ or $g'_{i,j}$ for a given instance $j$, we follow all of them simultaneously, obtaining subsets of $[n]$ instead of just single values in $[n]$.

Let us show a likely correspondence between the new instance of $\opintersect(\SC_{n,p})$
and the original instance of $\opor_t(\opequal(\PC_{n,p}))$. First, if the result of solving the instance of $\opor_t(\opequal(\PC_{n,p}))$ is 1, then clearly, by following 
the mappings from the instance of $\opequal(\PC_{n,p})$ resulting in 1, we also obtain an element that belongs to the intersection of two resulting sets in $\opintersect(\SC_{n,p})$.

Consider the case that the result of solving the instance of $\opor_t(\opequal(\PC_{n,p}))$ is 0. We bound the probability that the sets  
appearing in the instance of $\opintersect(\SC_{n,p})$ intersect. Each element of these two sets can be expressed as 
\[f'_{1,a_1}(f'_{2,a_2}(\ldots f'_{p,a_p}(1)\ldots))\]
or
\[g'_{1,b_1}(g'_{2,b_2}(\ldots g'_{p,b_p}(1)\ldots)),\]
respectively, where the sequences $a_1$, \ldots, $a_p$ and $b_1$, \ldots, $b_p$ describe
which of the instances the mapping is followed. There are $t^{2p}$ different pairs of such sequences. What is the probability that we obtain the same value for a specific pair of sequences? We want to show that this probability is bounded by $r^{p-1}/n$. 
If $a_1=\ldots=a_p = b_1= \ldots=b_p$, then we obtain different values, because the $a_1$-th instance in $\opor_t(\opequal(\PC_{n,p}))$ results in $0$. Suppose now that it is not the case that
$a_1=\ldots=a_p = b_1= \ldots=b_p$. If $a_1 \ne b_1$, then the probability that we obtain the same value is exactly $1/n$, because the final values are created by two independent random permutations $\pi_{1,a_1}$ and $\rho_{1,b_1}$.
If $a_1 = b_1$, let $k$ be the lowest number greater than $1$ such that $a_k \ne a_{k-1}$ or $b_k \ne b_{k-1}$. Since the functions $f'_{1,a_1}\circ\ldots\circ f'_{k-1,a_1}$ and 
$g'_{1,b_1}\circ\ldots\circ g'_{k-1,b_1}$ are not $r^{p-1}$-colliding 
and are applied to two values randomly distributed by 
$\pi_{k,a_k}$ and $\rho_{k,b_k}$, the probability of collision is at most $r^{p-1}/n$.
By the linearity of expectation, the expected size of the intersection between the two sets in the instance of $\opintersect(\SC_{n,p})$ is bounded by $t^{2p}\cdot r^{p-1}/n \le 1/10$. By Markov's inequality, the probability that the intersection is nonempty is bounded by $1/10$, so the probability that the reduction fails is bounded by $1/10$. Therefore, if we have a communication protocol for $\opintersect(\SC_{n,p})$ that errs with probability at most $1/10$, we can use this protocol to obtain a public-randomness protocol for $\opor_t(\opequal(\PC_{n,p}))$ that errs with probability at most $2/10$, provided no function in $\opor_t(\opequal(\PC_{n,p}))$ is $r$-colliding.
\jqed\end{proof}

\section{Proof of Main Tool (Theorem~\ref{thm:main})}
\label{sec:main}
We now combine the results of Steps 1, 2, and 3 to conclude our main communication complexity lower bound (Theorem \ref{thm:main} from Section \ref{sec:reductions}).

\ifnum\journal=0
\begin{proof}[Proof of Theorem~\ref{thm:main}]
\else
\begin{proof}[Theorem~\ref{thm:main}]
\fi
Let $r = C_\star(1+\log n)$ and let $t = \left\lfloor\frac{n^{1/(2p)}}{\sqrt{10r}}\right\rfloor$.
Due to the result of Newman~\cite{Newman91}, we know that every protocol with public randomness can be simulated using private randomness if we allow for using additional $O(\log(\mbox{input-size-in-bits}))$ communication bits and for increasing the probability of error by an arbitrarily small constant.
By combining this fact with Lemma~\ref{lem:step2} (which can be applied for $n$ greater than some constant), we find out that any public-randomness protocol for $\opor_t(\LEP_{n,p,r})$ that errs with probability at most $2/10$ has to use at least 
\[\Omega\left(\frac{tn}{p^{16}\log n}\right)- O\left(pt^2\right) - O(\log(t\cdot 2p \cdot n \cdot \log n))
= \Omega\left(\frac{tn}{p^{16}\log n}\right)- O\left(pt^2+\log n\right)
\] bits of communication. 
Note that for $n$ greater than some positive constant, the first term dominates the second, so we can express the lower bound as simply $\Omega\left(\frac{tn}{p^{16}\log n}\right)$.

Note that for our setting of $t$, $t^{2p} r^{p-1}\le n/10$. We can therefore apply  Lemma~\ref{lem:set_inters_vs_or_lep}. We learn that any communication protocol for 
$\opintersect(\SC_{n,p})$ that errs with probability at most $1/10$ has to use 
at least
$\Omega\left(\frac{tn}{p^{16}\log n}\right) - 2p$
bits of communication. As before, the first term dominates the second for sufficiently large $n$
and the lower bound becomes
\[ \Omega \left(\frac{n^{1+1/(2p)}}{p^{16}\cdot\log^{3/2}n}\right).\qedhere \]%
\jqed\end{proof}

\section{Acknowledgements}

We thank the anonymous reviewers for carefully reading the manuscript and many valuable comments.

\ifnum\journal=1
\bibliographystyle{spmpsci}
\else
\bibliographystyle{plain}
\fi
\bibliography{bibliography}

\begin{thebibliography}{10}

\bibitem{AhnG11}
Kook~Jin Ahn and Sudipto Guha.
\newblock Linear programming in the semi-streaming model with application to
  the maximum matching problem.
\newblock In {\em ICALP (2)}, pages 526--538, 2011.

\bibitem{BJKS04}
Ziv Bar{-}Yossef, T.~S. Jayram, Ravi Kumar, and D.~Sivakumar.
\newblock An information statistics approach to data stream and communication
  complexity.
\newblock {\em J. Comput. Syst. Sci.}, 68(4):702--732, 2004.

\bibitem{BBCR10}
Boaz Barak, Mark Braverman, Xi~Chen, and Anup Rao.
\newblock How to compress interactive communication.
\newblock {\em {SIAM} J. Comput.}, 42(3):1327--1363, 2013.

\bibitem{Braverman12}
Mark Braverman.
\newblock Interactive information complexity.
\newblock {\em {SIAM} J. Comput.}, 44(6):1698--1739, 2015.

\bibitem{BravermanR11}
Mark Braverman and Anup Rao.
\newblock Information equals amortized communication.
\newblock {\em {IEEE} Transactions on Information Theory}, 60(10):6058--6069,
  2014.

\bibitem{BravermanW12}
Mark Braverman and Omri Weinstein.
\newblock A discrepancy lower bound for information complexity.
\newblock In Anupam Gupta, Klaus Jansen, Jos{\'e} D.~P. Rolim, and Rocco~A.
  Servedio, editors, {\em APPROX-RANDOM}, volume 7408 of {\em Lecture Notes in
  Computer Science}, pages 459--470. Springer, 2012.

\bibitem{CCM08}
Amit Chakrabarti, Graham Cormode, and Andrew McGregor.
\newblock Robust lower bounds for communication and stream computation.
\newblock In {\em STOC}, pages 641--650, 2008.

\bibitem{ChakrabartiKW12}
Amit Chakrabarti, Ranganath Kondapally, and Zhenghui Wang.
\newblock Information complexity versus corruption and applications to
  orthogonality and {G}ap-{H}amming.
\newblock In {\em APPROX-RANDOM}, pages 483--494, 2012.

\bibitem{CSWY01}
Amit Chakrabarti, Yaoyun Shi, Anthony Wirth, and Andrew Chi-Chih Yao.
\newblock Informational complexity and the direct sum problem for simultaneous
  message complexity.
\newblock In {\em FOCS}, pages 270--278, 2001.

\bibitem{DJS98}
Carsten Damm, Stasys Jukna, and Jiri Sgall.
\newblock Some bounds on multiparty communication complexity of pointer
  jumping.
\newblock {\em Computational Complexity}, 7(2):109--127, 1998.

\bibitem{SGP11}
Atish {Das Sarma}, Sreenivas Gollapudi, and Rina Panigrahy.
\newblock Estimating {PageRank} on graph streams.
\newblock {\em J. ACM}, 58(3):13, 2011.

\bibitem{DGS87}
Pavol Duris, Zvi Galil, and Georg Schnitger.
\newblock Lower bounds on communication complexity.
\newblock {\em Inf. Comput.}, 73(1):1--22, 1987.

\bibitem{EggertKMS12}
Sebastian Eggert, Lasse Kliemann, Peter Munstermann, and Anand Srivastav.
\newblock Bipartite matching in the semi-streaming model.
\newblock {\em Algorithmica}, 63(1-2):490--508, 2012.

\bibitem{EpsteinLMS11}
Leah Epstein, Asaf Levin, Juli{\'a}n Mestre, and Danny Segev.
\newblock Improved approximation guarantees for weighted matching in the
  semi-streaming model.
\newblock {\em SIAM J. Discrete Math.}, 25(3):1251--1265, 2011.

\bibitem{FKMSZ05}
Joan Feigenbaum, Sampath Kannan, Andrew McGregor, Siddharth Suri, and Jian
  Zhang.
\newblock On graph problems in a semi-streaming model.
\newblock {\em Theor. Comput. Sci.}, 348(2-3):207--216, 2005.

\bibitem{FKMSZ08}
Joan Feigenbaum, Sampath Kannan, Andrew McGregor, Siddharth Suri, and Jian
  Zhang.
\newblock Graph distances in the data-stream model.
\newblock {\em SIAM J. Comput.}, 38(5):1709--1727, 2008.

\bibitem{GoelKK12}
Ashish Goel, Michael Kapralov, and Sanjeev Khanna.
\newblock On the communication and streaming complexity of maximum bipartite
  matching.
\newblock In {\em SODA}, pages 468--485, 2012.

\bibitem{GuhaM08}
Sudipto Guha and Andrew McGregor.
\newblock Tight lower bounds for multi-pass stream computation via pass
  elimination.
\newblock In {\em ICALP}, pages 760--772, 2008.

\bibitem{GuhaM09}
Sudipto Guha and Andrew McGregor.
\newblock Stream order and order statistics: Quantile estimation in
  random-order streams.
\newblock {\em SIAM J. Comput.}, 38(5):2044--2059, 2009.

\bibitem{conf}
Venkatesan Guruswami and Krzysztof Onak.
\newblock Superlinear lower bounds for multipass graph processing.
\newblock In {\em Proceedings of the 28th Conference on Computational
  Complexity, {CCC} 2013, Palo Alto, California, USA, 5-7 June, 2013}, pages
  287--298, 2013.

\bibitem{HRR}
Monika~Rauch Henzinger, Prabhakar Raghavan, and Sridhar Rajagopalan.
\newblock Computing on data streams.
\newblock Technical Report 1998-011, DEC System Research Center, 1998.

\bibitem{JPY12}
Rahul Jain, Attila Pereszl{\'e}nyi, and Penghui Yao.
\newblock A direct product theorem for bounded-round public-coin randomized
  communication complexity.
\newblock {\em CoRR}, abs/1201.1666, 2012.

\bibitem{JRS_substate}
Rahul Jain, Jaikumar Radhakrishnan, and Pranab Sen.
\newblock Privacy and interaction in quantum communication complexity and a
  theorem about the relative entropy of quantum states.
\newblock In {\em FOCS}, pages 429--438, 2002.

\bibitem{JRS}
Rahul Jain, Jaikumar Radhakrishnan, and Pranab Sen.
\newblock A direct sum theorem in communication complexity via message
  compression.
\newblock In {\em ICALP}, pages 300--315, 2003.

\bibitem{Kapralov13}
Michael Kapralov.
\newblock Better bounds for matchings in the streaming model.
\newblock In {\em SODA}, pages 1679--1697, 2013.

\bibitem{KLLRX12}
Iordanis Kerenidis, Sophie Laplante, Virginie Lerays, J{\'{e}}r{\'{e}}mie
  Roland, and David Xiao.
\newblock Lower bounds on information complexity via zero-communication
  protocols and applications.
\newblock {\em {SIAM} J. Comput.}, 44(5):1550--1572, 2015.

\bibitem{KonradMM12}
Christian Konrad, Fr{\'e}d{\'e}ric Magniez, and Claire Mathieu.
\newblock Maximum matching in semi-streaming with few passes.
\newblock In {\em APPROX-RANDOM}, pages 231--242, 2012.

\bibitem{McGregor05}
Andrew McGregor.
\newblock Finding graph matchings in data streams.
\newblock In {\em APPROX-RANDOM}, pages 170--181, 2005.

\bibitem{Newman91}
Ilan Newman.
\newblock Private vs. common random bits in communication complexity.
\newblock {\em Inf. Process. Lett.}, 39(2):67--71, 1991.

\bibitem{NisanW93}
Noam Nisan and Avi Wigderson.
\newblock Rounds in communication complexity revisited.
\newblock {\em SIAM J. Comput.}, 22(1):211--219, 1993.

\bibitem{PS84}
Christos~H. Papadimitriou and Michael Sipser.
\newblock Communication complexity.
\newblock {\em J. Comput. Syst. Sci.}, 28(2):260--269, 1984.

\bibitem{PRV01}
Stephen Ponzio, Jaikumar Radhakrishnan, and Srinivasan Venkatesh.
\newblock The communication complexity of pointer chasing.
\newblock {\em J. Comput. Syst. Sci.}, 62(2):323--355, 2001.

\bibitem{RazW89}
Ran Raz and Avi Wigderson.
\newblock Probabilistic communication complexity of {B}oolean relations
  (extended abstract).
\newblock In {\em FOCS}, pages 562--567, 1989.

\bibitem{Zelke12}
Mariano Zelke.
\newblock Weighted matching in the semi-streaming model.
\newblock {\em Algorithmica}, 62(1-2):1--20, 2012.

\end{thebibliography}

\end{document}